\documentclass[letterpaper]{article} % DO NOT CHANGE THIS
\usepackage{aaai2026}  % DO NOT CHANGE THIS
\usepackage{times}  % DO NOT CHANGE THIS
\usepackage{helvet}  % DO NOT CHANGE THIS
\usepackage{courier}  % DO NOT CHANGE THIS
\usepackage[hyphens]{url}  % DO NOT CHANGE THIS
\usepackage{graphicx} % DO NOT CHANGE THIS
\usepackage{natbib}  % DO NOT CHANGE THIS AND DO NOT ADD ANY OPTIONS TO IT
\usepackage{caption} % DO NOT CHANGE THIS AND DO NOT ADD ANY OPTIONS TO IT
\usepackage{algorithm}
\usepackage{algorithmic}
\usepackage{multirow}
\usepackage{booktabs}
\usepackage{pifont}
\usepackage{amsmath}
\usepackage{bm}
\usepackage{enumitem}
\usepackage{amssymb}
\usepackage{amsmath}% 加载 amsmath 宏包，用于数学公式
\usepackage{amsthm}% 加载 amsthm 宏包，用于定义定理环境
\newtheorem{lemma}{Lemma}% 引理与定理共享编号
\newtheorem{assumption}{Assumption}
\usepackage{makecell}
\usepackage{xcolor}
\definecolor{myred}{HTML}{e37777}
\definecolor{myblue}{HTML}{7faee3}
\usepackage{newfloat}
\usepackage{listings}
\DeclareCaptionStyle{ruled}{labelfont=normalfont,labelsep=colon,strut=off} % DO NOT CHANGE THIS
\floatstyle{ruled}
\newfloat{listing}{tb}{lst}{}
\floatname{listing}{Listing}
\nocopyright

\title{Breaking the Aggregation Bottleneck in Federated Recommendation:\\
A Personalized Model Merging Approach}
\author{
    Jundong Chen\textsuperscript{\rm 1,2},
    Honglei Zhang\textsuperscript{\rm 1,2},
    Chunxu Zhang\textsuperscript{\rm 3},
    Fangyuan Luo\textsuperscript{\rm 4},
    Yidong Li\textsuperscript{\rm 1,2}\thanks{Corresponding author.}
}
\affiliations{
    \textsuperscript{\rm 1}Key Laboratory of Big Data \& Artificial Intelligence in Transportation, Ministry of Education, China\\
    \textsuperscript{\rm 2}Beijing Jiaotong University\quad\textsuperscript{\rm 3}Jilin University\quad\textsuperscript{\rm 4}Beijing University of Technology
    \{jundongchen, honglei.zhang, ydli\}@bjtu.edu.cn, cxzhang19@mails.jlu.edu.cn, luofangyuan@bjut.edu.cn
}

\usepackage{bibentry}
\begin{document}

\maketitle

\begin{abstract}
% Federated recommendation (FR) enables collaborative training via a central server, without uploading raw user data, thereby ensuring privacy and enabling client-specific recommendations. 
Federated recommendation (FR) facilitates collaborative training by aggregating local models from massive devices, enabling client-specific personalization while ensuring privacy. However, we empirically and theoretically demonstrate that server-side aggregation can undermine client-side personalization, leading to suboptimal performance, which we term the aggregation bottleneck. This issue stems from the inherent heterogeneity across numerous clients in FR, which drives the globally aggregated model to deviate from local optima. To this end, we propose FedEM, which elastically merges the global and local models to compensate for impaired personalization. Unlike existing personalized federated recommendation (pFR) methods, FedEM (1) investigates the aggregation bottleneck in FR through theoretical insights, rather than relying on heuristic analysis; (2) leverages off-the-shelf local models rather than designing additional mechanisms to boost personalization. Extensive experiments on real-world datasets demonstrate that our method preserves client personalization during collaborative training, outperforming state-of-the-art baselines. Our code is available at \url{https://github.com/jundongchen13/FedEM}.
% Our code and appendix are provided in the supplementary materials.
\end{abstract}

% Our code and appendix are available at \url{https://anonymous.4open.science/r/FedEM}.


\section{Introduction}
Federated recommendation (FR), as an emerging on-device learning paradigm, ensures that clients' raw data remains local during the training process, thus protecting user privacy \cite{sun2022survey,yin2024device}. Pioneering works include FedMF~\cite{chai2020secure} and FedNCF~\cite{perifanis2022federated}, which apply matrix factorization and neural collaborative filtering, respectively, within the federated framework. Conceptually, the clients upload local models after local training to the server for global aggregation. Later, they download the global model as the new local model for the next training round~\cite{chen2023win}.

\begin{figure}[t]
\centering
\includegraphics[width=0.82\columnwidth]{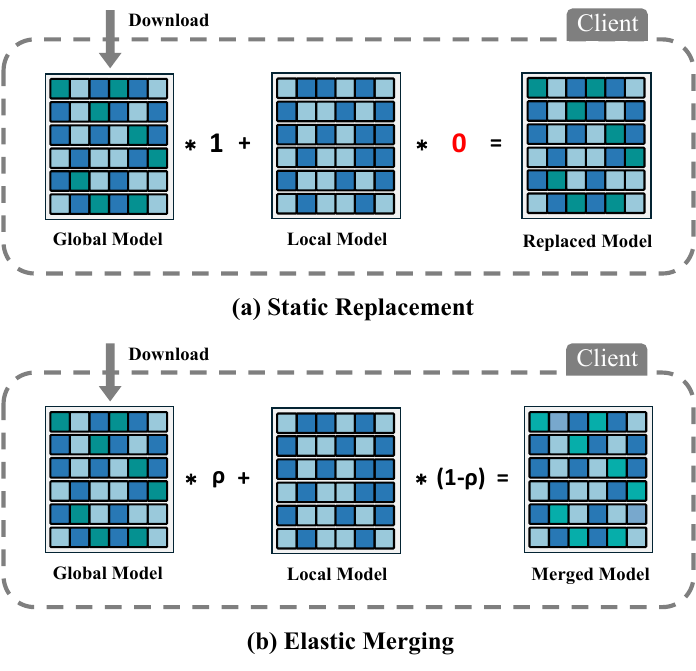} 
\caption{Traditional FR methods directly replace the local model with the global model downloaded from the server, while our FedEM can elastically merge both the global model and local model with the weight $\rho$, thereby delicately balancing collaboration and personalization.}
\label{pic:sr_and_em}
\end{figure}

\begin{figure*}[ht]
\centering
\includegraphics[width=0.99\textwidth]{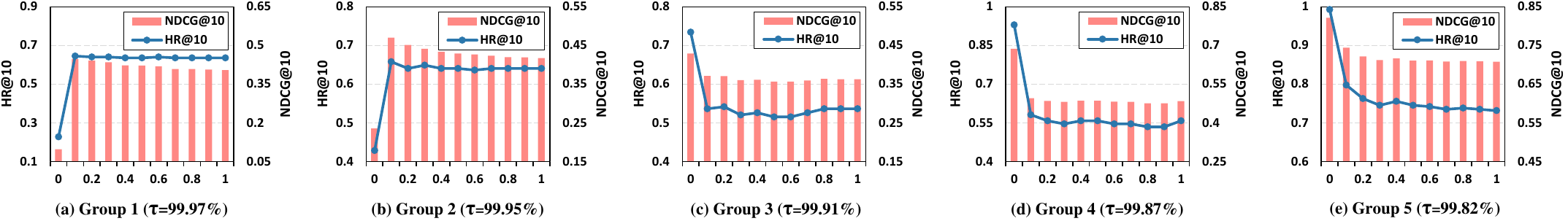} % Reduce the figure size so that it is slightly narrower than the column.
\caption{The performance of different client groups under varying merging weights $\rho$. The degree of data sparsity $\tau$ decreases from group 1 to group 5. Here, $\rho=1$ stands for the traditional static replacement scheme.}
\label{pic:pre_exp}
\vspace{-0.05cm}
\end{figure*}
% The x-axis indicates the weight $\rho$ assigned to the local model, thus $1-\rho$ to the global model. Equally, $\rho=0$ for traditional static replacement while $\rho=1$ for complete local learning without global aggregation.

However, due to the varying user preferences in FR, the interaction data from different clients are not independent and identically distributed (Non-IID), which naturally leads to data heterogeneity~\cite{sun2022survey}. Traditional FR methods fail to address this issue by sharing one same global model across all clients~\cite{lightfr_2022}. Hence, personalized federated recommendation (pFR) has been introduced to tailor client-specific models. For example, PFedRec~\cite{zhang2023dual} enables dual personalization of both local and global components, while FedRAP~\cite{li2023federated} trains an extra personalized model locally. Although effective in practice, existing pFR methods suffer from two key limitations: (1) They overlook the degradation of local personalization caused by global aggregation in FR; (2) They rely on heuristically designed personalization mechanisms, which limit their compatibility.
% pFedGraph \cite{ye2023personalized} infers the collaboration graph to optimize personalized models for each client.

For the first limitation, we theoretically reveal that aggregation can cause a loss of local information in the global model. This issue is exacerbated by the unique nature of FR tasks, which often involve millions of clients, several orders of magnitude more than typical federated learning scenarios. Besides, existing FR methods typically adopt a \texttt{static replacement} scheme, as illustrated in Figure~\ref{pic:sr_and_em}(a), where the global model replaces the local one for both training and inference. Such a scheme propagates the impact of aggregation, causing optimization to deviate from the client-specific optimum and ultimately undermining local personalization. We refer to this issue as the aggregation bottleneck.

% 一般是先发现了表象，再深入探究理论。这样的逻辑比较合理。上边的逻辑先总结一个理论，然后再探究现象不太合理。另外，单纯看上边的逻辑，先理论后现象，中间的besides连接词用的也不好，这里应该是说完理论之后想探究现象，所以可以用进一步的，表明要深挖这一理论背后的机制，这一的递进会更好，而besides更多的是并列。

% For the first limitation, we empirically observe that existing FR models adopt a \texttt{static replacement} scheme, as depicted in Figure~\ref{pic:sr_and_em}(a), where the global model entirely replaces the local one for both training and inference. Such a scheme amplifies aggregation effects, leading to deviate from client-specific optimum and reduced personalization. Building on this finding, we further theoretically prove that aggregation can compromise local information in the global model. This issue is exacerbated by the unique nature of FR tasks, which often involve millions of clients, several orders of magnitude more than typical federated learning scenarios. We refer to this issue as the aggregation bottleneck.

For the second limitation, we explore a simpler and more general solution. Grounded in model merging theory~\cite{zhouemergence, yang2024model}, we aim to compensate for local information loss with the off-the-shelf local model, thereby breaking the aggregation bottleneck. Specifically, we adopt an \texttt{elastic merging} scheme, as illustrated in Figure~\ref{pic:sr_and_em}(b), where $\rho$ denotes the weight of the global model parameters, while $1-\rho$ stands for the weight of the local ones. We conduct empirical validation based on FedMF. To account for client heterogeneity, we divide the clients into five groups according to privacy-insensitive statistics of their local data, \textit{e.g.}, data sparsity. For each group, we report the average performance metrics (HR@10 and NDCG@10) under different $\rho$. As shown in Figure~\ref{pic:pre_exp}, we observe that: (1) Static replacement scheme, \textit{i.e.}, $\rho=1$, yields suboptimal performance, which is consistent with our theoretical analysis; (2) Elastic merging scheme effectively enhances model performance. Owing to heterogeneity across clients, each group achieves optimal trade-offs under different merging weights $\rho$. Given this, we adopt client-specific merging that elastically adjusts to local personalization needs. 

Taking both limitations into account, we propose a simple yet theoretically guaranteed pFR method called \textbf{Fed}erated recommendation via \textbf{E}lastic \textbf{M}erging (\textbf{FedEM}). It merges the aggregated global model with the off-the-shelf local model in a balanced way, effectively absorbing collaborative information while preserving client-specific personalization. To sum up, our main contributions are as follows:

\begin{itemize}[left=0pt]
    \item To the best of our knowledge, we are the first to theoretically analyze the aggregation bottleneck in FR scenarios, \textit{i.e.}, the loss of local information caused by global aggregation, which further harms local personalization and leads to suboptimal performance.
    \item Based on model merging, we propose FedEM to bridge the gap between global collaboration and local personalization. Unlike existing heuristic methods, our approach is theoretically grounded and empirically validated.
    \item The proposed elastic merging module is model-agnostic and can be seamlessly integrated as a plug-in to enhance FR/pFR methods constrained by aggregation bottleneck.
    \item Extensive experiments demonstrate that FedEM consistently outperforms state-of-the-art (SOTA) methods.
\end{itemize}

\section{Related Work}
\subsection{Personalized Federated Learning} 
Federated learning (FL) is a distributed learning paradigm to mitigate privacy issues~\cite{huang2024federated, guo2025federated}. Traditional FL methods, such as FedAvg~\cite{mcmahan2017communication}, struggle to derive a global model generalized for each client when the local data is Non-IID~\cite{huang2021personalized}. To that end, some personalized federated learning (pFL) methods aim to fine-tune the globally aggregated model for each client to obtain the personalized ones~\cite{collins2021exploiting,fallah2020personalized}. For instance, FedALA exploits the general information in the lower layers of the global model to enhance the capability of feature extraction of the local model targeting federated vision tasks~\cite{zhang2023fedala}; Still other methods tend to locally learn a personalized model for each client~\cite{li2021ditto},~\textit{e.g.}, pFedMe~\cite{t2020personalized} uses Moreau envelopes as the client regulation loss to decouple personalized model optimization from the global model learning; Other methods take a global view of personalized aggregation for each client~\cite{ji2019learning, zhang2020personalized}, \textit{e.g.}, pFedGraph~\cite{ye2023personalized} optimizes the collaboration graph to perform weighted aggregation for different clients.

\subsection{Personalized Federated Recommendation} 
In federated recommendation (FR), there is natural heterogeneity among clients due to their different preferences, thus necessitating personalization of the local model~\cite{sun2022survey,yin2024device}. Similar to pFL, personalized federated recommendation (pFR) methods can also be categorized into three research lines as follows. (1) Fine-tune the global model~\cite{zhang2024transfr}: PFedRec~\cite{zhang2023dual} personalizes both the score function and global model to alleviate data heterogeneity. FedCIA~\cite{han2025fedcia} leverages the aggregated item similarity matrix to guide local model training, aiming to achieve global fine-tuning. (2) Learn an additional personalized model~\cite{chen2025beyond}: FedRAP~\cite{li2023federated} applies an additive model to item embeddings to enhance personalization. (3) Perform personalized global aggregation~\cite{luo2022personalized, zhang2024beyond}: FedFast~\cite{muhammad2020fedfast} performs global aggregation by identifying representatives from different clusters based on user profile similarities. GPFedRec~\cite{zhang2024gpfedrec} performs graph-guided aggregation to recover inter-client correlations, generating client-specific global models. Unlike the above methods, which follow the traditional static replacement scheme and rely on heuristic personalization mechanisms, FedEM adopts an elastic merging scheme, offering a simple yet theoretically grounded solution. During global collaborative training, it leverages the off-the-shelf local model to enhance personalization for each client.

\section{Preliminaries}
% \subsection{Problem Formulation}

Let $\mathcal{U}$ denote the set with $n$ users/clients, $\mathcal{I}$ the set with $m$ items. Then $\mathcal{D}_{u} = \{(u,i,r_{ui} | i\in\mathcal{I}_u)\}$ denotes the local interaction dataset of client $u$, where $\mathcal{I}_u$ is for the items observed by client $u$, and each entry $r_{ui}\in \{0,1\}$ is for the label on item $i$ by client $u$. The goal of FR is to predict $\hat{r}_{ui}$ of client $u$ for each unobserved item $i \in \mathcal{I} \setminus \mathcal{I}_u$ on local devices. Formally, the global optimization objective over $n$ clients of FR tasks is:
\begin{equation}\label{eq:frs_fl}
\min _{\left(\mathbf{p}_1, \mathbf{p}_2,\cdots, \mathbf{p}_n; \mathbf{Q}_1^l, \mathbf{Q}_2^l,\cdots, \mathbf{Q}_n^l\right)} \sum_{u=1}^n p_u \mathcal{L}_u\left(\mathbf{p}_u, \mathbf{Q}_u^l ; \mathcal{D}_u\right),
\end{equation}
\noindent where $\mathbf{p}_u\in\mathbb{R}^d$ and $\mathbf{Q}_u^l\in\mathbb{R}^{m \times d}$ denote local user embedding and local item embedding table for client $u$, respectively, $d$ is the dimension of the embedding vector. In this work, we treat item embedding table $\mathbf{Q}$ as the intermediate model parameters for uploads and downloads, since it is the standard configuration in embedding-based FRs \cite{ammad2019federated}. The server aggregates local parameters $\mathbf{Q}_u^l$ with weight $p_u$ to derive global model $\mathbf{Q}^g$, \textit{e.g.}, $p_u=|\mathcal{D}_u|/ \sum_{v=1}^n|\mathcal{D}_v|$ in FedAvg \cite{mcmahan2017communication} and $p_u=1/n$ in FCF~\cite{ammad2019federated}. $\mathcal{L}_u$ is the task-specific objective to facilitate local training.

In this work, we focus on the typical recommendation task with implicit feedback, where the rating $r_{ui}$ of item $i$ by user $u$ is either 1 or 0, indicating interested versus uninterested interaction, respectively. Therefore, we adopt the binary cross-entropy (BCE) loss \cite{he2017neural}, which is commonly used in binary-value problems, as the local objective $\mathcal{L}_u$. Formally, the BCE loss is defined as:
\begin{equation}\label{eq:bce}
\begin{split}
    \mathcal{L}_{u}(\mathbf{p}_u, \mathbf{Q}_u^l ; & \mathcal{D}_u) =-\sum_{(u, i, r_{u i}) \in \mathcal{D}_u} \big [r_{u i} \log \hat{r}_{u i}\\& +\left(1-r_{u i}\right) \log \left(1-\hat{r}_{u i}\right)\big ],
\end{split}
\end{equation}
\noindent where $\hat{r}_{u i}=\phi(\mathbf{p}_u\mathbf{q}_i^{\top})$ denotes the predicted rating of item $i$ by user $u$, where $\mathbf{q}_i\in\mathbf{Q}_u^l$ and $\phi(\cdot)$ is the Sigmoid function.

\begin{figure*}[htbp]
\centering
\includegraphics[width=0.99\textwidth]{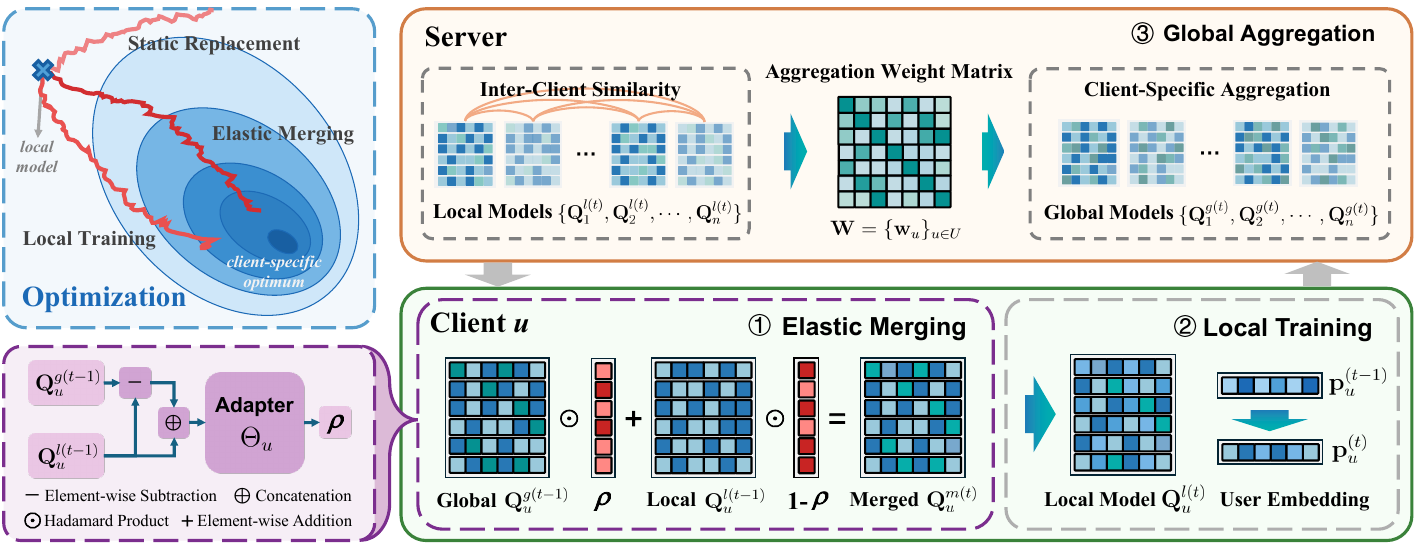} % Reduce the figure size so that it is slightly narrower than the column. Don't use precise values for figure width.This setup will avoid overfull boxes.
\caption{The framework of FedEM. Unlike the static replacement scheme in existing FR methods, we propose an elastic merging scheme to mitigate the optimization deviation caused by global aggregation, effectively breaking the performance bottleneck. Additionally, the server can perform similarity-based aggregation to further alleviate the aggregation bottleneck.}
\label{pic:framework}
\end{figure*}

\section{Methodology}
\subsection{Motivation}
\label{sec:analysis}
\textbf{Aggregation Bottleneck.} Most existing FR methods adopt fixed-weight aggregation on the server side to enable collaborative training across clients. Taking FCF~\cite{ammad2019federated} for example, that is, $p_u=1/n$, the global model derived by aggregation in round $t$ can be represented as:

\begin{equation}\label{eq:agg}
    \begin{split}
        \mathbf{Q}^{g(t)} & = \frac{1}{n}\sum_{u=1}^{n}\mathbf{Q}_u^{l(t)}=\frac{1}{n} \mathbf{Q}_u^{l(t)} + \frac{n - 1}{n} \mathbf{Q}_r^{g(t)}, \\ \quad & \text{where} \quad \mathbf{Q}_r^{g(t)} = \frac{1}{n - 1} \sum_{v \ne u}^n \mathbf{Q}_v^{l(t)}.
    \end{split}
\end{equation}
Here, $\mathbf{Q}_r^{g(t)}$ denotes the aggregated model over the remaining clients except for ego Client $u$. However, unlike typical FL settings, FR tasks often involve tens of thousands or even millions of clients. Therefore, we have:
\begin{equation}
\label{eq:lim}
\lim_{n\to\infty}\frac{1}{n}=0,\lim_{n\to\infty}\frac{n-1}{n}=1,
\end{equation}
which means $\mathbf{Q}^{g(t)} \approx \mathbf{Q}_r^{g(t)}$, indicating a loss of local information in the global model. At this point, for client $u$, adopting the static replacement scheme may lead to the aggregation bottleneck formally established in Lemma~\ref{lemma:misalignment}.
\begin{assumption}[Smoothness and Strong Convexity]
\label{assumption:smooth_strong}
The local objective function $\mathcal{L}_u(\mathbf{Q})$ is assumed to be continuously differentiable, $L$-smooth, and $\mu$-strongly convex with respect to the model parameter $\mathbf{Q}$. This assumption can be approximately satisfied in our setting by incorporating $\ell_2$ regularization into the BCE loss~\cite{boyd2004convex, bottou2018optimization}. Such assumptions are also standard in theoretical analyses of federated learning~\cite{mcmahan2017communication, t2020personalized}. 

% In practice, this assumption can be satisfied by incorporating $\ell_2$ regularization into the BCE loss. All subsequent theoretical analyses are developed based on this assumption.
\end{assumption}
\begin{lemma}[Performance Bottleneck of Global Aggregation]\label{lemma:misalignment}
    Based on Assumption~\ref{assumption:smooth_strong}, let $\mathbf{Q}_u^*$ denote the local optimal model. Then, following the static replacement scheme, i.e., using the global model $\mathbf{Q}^{g(t)}$ for local training, may lead to:
    \begin{equation}
        \langle \nabla \mathcal{L}_u(\mathbf{Q}^{g(t)}), \mathbf{Q}^{g(t)} - \mathbf{Q}_u^* \rangle \leq 0,
    \end{equation}
    which shows that the globally aggregated model can cause training optimization to deviate from the client-specific optimum, thereby undermining local personalization. This issue becomes more pronounced under the high inherent client heterogeneity in FR, leading to a performance bottleneck.
\end{lemma}

% $p_u=|\mathcal{D}_u|/ \sum_{v=1}^n|\mathcal{D}_v|$
\subsection{Framework}
As illustrated in Figure \ref{pic:framework}, the overall framework of FedEM consists of the following main steps in each round $t$: 
\\ \ding{182} Elastic Merging: client $u$ merges aggregated global model $\mathbf{Q}_u^{g(t-1)}$ and off-the-shelf local model $\mathbf{Q}_u^{l(t-1)}$ to derive a merged model $\mathbf{Q}_u^{m(t)}$ at the beginning of this round $t$. Here, existing FR methods follow the static replacement scheme and directly discard the local model.
\\ \ding{183} Local Training: client $u$ utilizes its own data $\mathcal{D}_{u}$ to update the merged model $\mathbf{Q}_u^{m(t)}$ and local parameters like $\mathbf{p}_u^{(t-1)}$. Then, client $u$ uploads the updated model $\mathbf{Q}_u^{l(t)}$ to the server.\\ \ding{184} Global Aggregation: the server performs aggregation based on the similarity over uploaded $\{\mathbf{Q}_1^{l(t)}, \cdots, \mathbf{Q}_n^{l(t)}\}$, generating the global model $\mathbf{Q}_u^{g(t)}$ for client $u$. 
\\\noindent Notably, both $\mathbf{Q}_u^{g(0)}$ and $\mathbf{Q}_u^{l(0)}$ represent randomly initialized models when $t=1$. After $T$ rounds above, each client can get its own personalized model. 

\subsection{Elastic Merging} 
% \subsection{Theoretical Analysis: Model Merging}
\textbf{Theoretical Analysis.} Existing pFR methods, \textit{e.g.}, PFedRec~\cite{zhang2023dual}, FedRAP~\cite{li2023federated}, and FedCIA~\cite{han2025fedcia}, heuristically design additional personalization mechanisms. Although empirically effective, these methods fail to identify the aggregation bottleneck as the underlying cause of insufficient personalization in FR, thereby limiting their potential for further improvement. Additionally, their complexity and limited compatibility pose challenges for practical deployment. 

Inspired by the model merging theory~\cite{zhouemergence, yang2024model}, we merge the global model $\mathbf{Q}^{g(t)}$ and the off-the-shelf local model $\mathbf{Q}_u^{l(t)}$ to overcome the aggregation bottleneck. For embedding-based FRs, $\forall\rho\in[0,1]$,
\begin{equation}\label{eq:merging}
    f(\rho\mathbf{Q}^{g(t)}+(1-\rho)\mathbf{Q}_u^{l(t)})=\rho f(\mathbf{Q}^{g(t)})+(1-\rho)f(\mathbf{Q}_u^{l(t)}),
\end{equation}
where $f(\mathbf{Q})$ denotes the predicted rating, \textit{i.e.}, $\mathbf{p}_u\cdot\mathbf{Q}$. Thus, the two models $\mathbf{Q}^{g(t)}$ and $\mathbf{Q}_u^{l(t)}$ satisfy the model merging conditions in the parameter space~\cite{zhouemergence}. We define the merged model as:
\begin{equation}
\label{eq:merge}
\mathbf{Q}_u^{m(t+1)} = \rho \mathbf{Q}^{g(t)} + (1 - \rho) \mathbf{Q}_u^{l(t)},
\end{equation}
which is used as the initialization for local training in round $t+1$. By following the elastic merging scheme, we compensate for the loss of local information in the aggregated model. The merged model benefits from global collaborative training while mitigating the risk of deviating from the client-specific optimum, thus preserving personalization.

\begin{lemma}[Compensation Effect of Elastic Merging]\label{lemma:alignment}
Based on Assumption~\ref{assumption:smooth_strong}, the merged model can compensate for the optimization deviation caused by aggregation, guiding the update toward the client-specific optimum, \textit{i.e.}, it satisfies:
\begin{equation}
\begin{split}
&\langle \nabla \mathcal{L}_u(\mathbf{Q}_u^{m(t+1)}),\, \mathbf{Q}_u^{m(t+1)} - \mathbf{Q}_u^* \rangle\\
&\geq (1 - \rho)\langle \nabla \mathcal{L}_u(\mathbf{Q}_u^{l(t)}),\, \mathbf{Q}_u^{l(t)} - \mathbf{Q}_u^* \rangle - \rho C,
\end{split}
\end{equation}
where $C$ denotes the deviation introduced by global aggregation. This deviation can be lower-bounded by both the state of the local model $\mathbf{Q}_u^{l(t)}$ and the global-local discrepancy $\mathbf{Q}_u^{\Delta(t)} = \mathbf{Q}^{g(t)} - \mathbf{Q}_u^{l(t)}$. The first term on the right-hand side of the inequality, $A = \langle \nabla \mathcal{L}_u(\mathbf{Q}_u^{l(t)}),\, \mathbf{Q}_u^{l(t)} - \mathbf{Q}_u^* \rangle>0$, reflects the preservation of client-specific optimization information. For effective local personalization, the update must remain aligned with the client-specific optimum, which requires the left-hand side of the inequality to be positive. To achieve this, we need to adaptively control the merging weight $ 0 \leq \rho < \frac{A}{A + C}$ based on the deviation $C$ in each round. Accordingly, we propose the Elastic Merging module, which learns $\rho$ from $\mathbf{Q}_u^{l(t)}$ and $\mathbf{Q}_u^{\Delta(t)}$, ensuring client-wise alignment and mitigating the aggregation bottleneck.
\end{lemma}
\noindent Detailed proofs of all lemmas are deferred to \textit{\textbf{Appendix A}}.

\noindent \textbf{Practical Implementation.} On the client side, we introduce the Elastic Merging (EM) module to dynamically preserve personalized information from the local model during global collaborative training. Furthermore, we utilize the vector $\bm{\rho} \in \mathbb{R}^m$ for fine-grained merging at the item level. Based on the analysis in Lemma~\ref{lemma:alignment}, we first compute the global-local model discrepancy $\mathbf{Q}_u^{\Delta(t-1)} = \mathbf{Q}_u^{g(t-1)} - \mathbf{Q}_u^{l(t-1)}$, and then concatenate it with the local model $\mathbf{Q}_u^{l(t-1)}$ to obtain $\mathbf{Q}_u^{\text{cat}(t-1)} = \mathbf{Q}_u^{\Delta(t-1)} \oplus \mathbf{Q}_u^{l(t-1)}$. We feed $\mathbf{Q}_u^{\text{cat}(t-1)}\in\mathbb{R}^{m\times2d}$ into the core component of the EM module, \textit{i.e.}, the adapter, to generate the elastic vector $\bm{\rho}$ for each client $u$. 

Here, we implement the adapter with an $L$-layer multilayer perceptron (MLP). Formally, we have:
\begin{equation}
\label{eq:mlp_output}
\bm{\rho}=\phi_{output}(\Theta_u^L\bullet(\cdots\phi(\Theta_u^2\bullet(\phi(\Theta_u^1\bullet\mathbf{Q}_u^{\text{cat}(t-1)})))),
\end{equation}
\noindent where $\phi(\cdot)$ and $\phi_{output}(\cdot)$ denote the mapping functions for the intermediate and output layers, respectively. Here, we employ ReLU for $\phi$ and Sigmoid for $\phi_{output}$ to ensure that each element of vector $\bm{\rho}$ is in the range $[0,1]$. Besides, notation $\bullet$ denotes the matrix product operation, and $\Theta_u^i$ denotes the parameters of the $i$-th layer. Then, the merged model can be formulated as:
% \begin{equation}
%     \mathbf{Q}_u^{m(t)} = \mathbf{Q}_u^{l(t-1)} + \bm{\rho} \odot \mathbf{Q}_u^{\Delta(t-1)},
% \end{equation}
\begin{equation}
\begin{split}
    \mathbf{Q}_u^{m(t+1)} & = \bm{\rho} \odot \mathbf{Q}_u^{g(t-1)} + (\mathbf{1} - \bm{\rho}) \odot \mathbf{Q}_u^{l(t-1)} \\& = \mathbf{Q}_u^{l(t-1)} + \bm{\rho} \odot \mathbf{Q}_u^{\Delta(t-1)},
\end{split}
\end{equation}
where $\odot$ denotes the Hadamard product. Hence, by preserving the local information in $\mathbf{Q}_u^{l(t)}$, $\bm{\rho}$ serves to selectively absorb the collaborative information in $\mathbf{Q}_u^{\Delta(t)}$, thus preventing the aforementioned aggregation bottleneck. 

Specifically, the adapter parameters $\Theta_u$ can be updated using stochastic gradient descent (SGD) as follows:

\begin{equation}
\label{eq:elastic_adaptation}
\Theta_u=\Theta_u-\beta \cdot \frac{\partial \mathcal{L}_{u}(\mathbf{p}_u^{(t-1)}, \mathbf{Q}_u^{m(t)} ; \mathcal{D}_u)}{\partial \Theta_u},
\end{equation}

\noindent where $\beta$ is the learning rate for the adapter. Notably, only the parameters $\Theta_u$ are updated with $\bm{\rho}$ varied in this step, while other model parameters, \textit{i.e.}, $\mathbf{p}_u^{(t-1)}$, $\mathbf{Q}_u^{l(t-1)}$ and $\mathbf{Q}_u^{g(t-1)}$, are frozen. Unlike other FR methods, $\mathbf{Q}_u^{m(t)}$ is utilized to initialize the local model for local training rather than $\mathbf{Q}_u^{g(t-1)}$.

%Remarkably, instead of updating the parameters of EA module once per iteration, we perform multiple rounds of updates in the same iteration simultaneously with the local training. so as to enhance the adaptability of EA module.

\subsection{Local Training} 
The objective of local training is to update the user embedding $\mathbf{p}_u^{(t-1)}$ and the item embedding table $\mathbf{Q}_u^{m(t)}$ to obtain $\mathbf{p}_u^{(t)}$ and $\mathbf{Q}_u^{l(t)}$. Note that the adapter $\Theta_u$ is not involved during this process. Formally, the update can be written as:
\begin{equation}
\label{eq:local_training}
\mathbf{p}_u=\mathbf{p}_u-\eta \cdot \frac{\partial \mathcal{L}_{u}}{\partial \mathbf{p}_u},\quad \mathbf{Q}_u^l=\mathbf{Q}_u^l-\eta \cdot \frac{\partial \mathcal{L}_{u}}{\partial \mathbf{Q}_u^l},
\end{equation}
\noindent where $\eta$ is the learning rate for local training. To protect user privacy, after local training, client $u$ only uploads $\mathbf{Q}_u^{l(t)}$ to the server for global aggregation.

\subsection{Global Aggregation} 
% After receiving the uploaded local models $\{\mathbf{Q}_u^{l(t)} \mid u \in \mathcal{U} \}$, 
As shown in Lemma~\ref{lemma:misalignment}, traditional fixed-weight aggregation is ineffective in FR, and the aggregation bottleneck is further exacerbated by inherent client heterogeneity.  In addition to local elastic merging, we further perform global similarity-based aggregation to alleviate the impact of such heterogeneity. Inspired by the previous work \cite{ye2023personalized} in FL, we additionally incorporate client similarity, generating a client-specific weight vector $\mathbf{w}_u=[w_{u1}, w_{u2}, \dots, w_{un}]\in\mathbb{R}^n$ to conduct tailored aggregation for each client $u$. Specifically, $\mathbf{w}_u$ is obtained by optimizing the following objective:
\begin{equation}
\label{eq:optimize_w}
\begin{split}
\mathcal{L}_g = \sum_{v=1}^n \left((w_{uv}-p_v)^2 + \alpha (w_{uv}-\sigma(\mathbf{Q}_u^{l},\mathbf{Q}_v^{l}))^2\right),
%& \text{s.t.} \quad \mathbf{1}^T \mathbf{w}_u = 1; \,\,\, \mathbf{w}_u > \mathbf{0}.
\end{split}
\end{equation}
\noindent where $p_u = |\mathcal{D}_u| / \sum_{v=1}^n |\mathcal{D}_v|$ denotes the aggregation weight used in the backbone FedMF~\cite{chai2020secure}. Here, $\sigma(\mathbf{Q}_u^l, \mathbf{Q}_v^l)=1/(1+\|\mathbf{Q}_u^l-\mathbf{Q}_v^l\|^2)$ denotes the similarity function. This term ensures that the aggregation weight $w_{uv}$ increases when the two client models are highly similar. Besides, $\alpha$ is a similarity-related hyperparameter. After global aggregation, each client $u$ can download the global model $\mathbf{Q}_u^{g(t)}=\sum_{v=1}^n w_{uv}\mathbf{Q}_v^{l(t)}$ to start the next round $t+1$. 

\noindent The detailed procedure of our algorithm is in \textit{\textbf{Appendix B}}.

% It can also be substituted by other similarity functions,~\textit{e.g.}, cosine similarity.

\section{Discussions}
\subsection{Complexity Analysis}
Given $n$ clients and $m$ items, with embedding dimension $d$, the computational complexity of the local backbone model is $\mathcal{O}((m + 1)d)$. Our proposed EM module is implemented as an $L$-layer MLP, whose complexity is $\mathcal{O}(Ld^2)$. Since $m \gg d > L$ in practice, the EM module introduces negligible overhead to the local device, making it well-suited for on-device recommendation.

% Additionally, the server computes client similarity with a complexity of $\mathcal{O}(n^2md)$, which imposes no burden on local clients. Therefore, our method is well-suited for on-device recommendation scenarios.

\subsection{Privacy Analysis}
FedEM naturally preserves user privacy under the federated paradigm. Moreover, the local merging parameters $\Theta_u$ and $\bm{\rho}$ are not shared with the server, reducing the risk of privacy leakage~\cite{chai2020secure}. To further enhance privacy, we incorporate local differential privacy (LDP)~\cite{qi2024towards} into our method. The privacy budget $\varepsilon$ is guaranteed by $\mathcal{S}_u / \delta$, where $\delta$ is the noise strength and $\mathcal{S}_u$ denotes the global sensitivity of client $u$. We upper-bound $\mathcal{S}_u$ by $2p_u\eta Z$, where $Z$ is the gradient clipping threshold.

% through analysis of consecutive model 

\noindent Detailed discussions can be found in \textit{\textbf{Appendix C}}.

\section{Experiments}
\subsection{Experimental Settings}
\textbf{Datasets.} We evaluate our proposed method on four datasets with varying client scale and data sparsity: Filmtrust~\cite{filmtrust_2013}, ML-100K \cite{movielens_2015}, ML-1M~\cite{movielens_2015}, and LastFM-2K~\cite{cantador2011second}.\\
\noindent\textbf{Evaluation Protocols.} We follow the popular \textit{leave-one-out} evaluation \cite{bayer2017generic, he2017neural}, and report the performance by \textit{Hit Ratio} (HR@$K$) and \textit{Normalized Discounted Cumulative Gain} (NDCG@$K$) \cite{he2015trirank}.\\
\noindent\textbf{Compared Baselines.} We compare FedEM with SOTA centralized methods, \textit{e.g.}, MF~\cite{koren2009matrix}, NCF~\cite{he2017neural}, LightGCN~\cite{he2020lightgcn}, and federated methods, \textit{e.g.}, FedMF~\cite{chai2020secure}, FedNCF~\cite{perifanis2022federated}, FedFast~\cite{muhammad2020fedfast}, PFedRec~\cite{zhang2023dual}, CoLR~\cite{nguyen2024towards}, GPFedRec~\cite{zhang2024gpfedrec}, FedRAP~\cite{li2023federated}, FedCIA~\cite{han2025fedcia}.

\noindent\textbf{Implementation Details.} For a fair comparison, we set the global round $T=100$, batch size $B=256$, and embedding dimension $d=16$ for all methods. Besides, we adopt the default/optimal settings for other hyperparameters as reported in their original papers. All baseline methods reach convergence under the given settings.\\
\noindent Detailed experimental settings are provided in \textit{\textbf{Appendix D}}.

% Based on the results, we have the following observations and insights:
\subsection{Overall Performance}\label{sec:main_exp}

\begin{table*}[htbp]
\centering
\resizebox{\textwidth}{!}{
\Large
\begin{tabular}{llcccccccccccc}
\toprule
\multirow{2}{*}{\textbf{Datasets}} & \multirow{2}{*}{\textbf{Methods}}& \multicolumn{3}{c}{\textbf{CenRec}} & \multicolumn{8}{c}{\textbf{FedRec}} & \multicolumn{1}{c}{\textbf{Ours}} \\
\cmidrule(lr){3-5}\cmidrule(lr){6-13}\cmidrule(lr){14-14}& & \textbf{MF} & \textbf{NCF} & \textbf{LightGCN} & \textbf{FedMF} & \textbf{FedNCF} & \textbf{FedFast} & \textbf{PFedRec}& \textbf{CoLR} & \textbf{GPFedRec} & \textbf{FedRAP} & \textbf{FedCIA} & \textbf{FedEM} \\
\midrule
\multirow{2}{*}{\textbf{Filmtrust}} & HR@10 &0.6936 & 0.6856 & 0.7804 & 0.6577 & 0.6597 & 0.6637 & 0.6756 & 0.6377 & 0.6836 & 0.8024 & \underline{0.8543} & \textbf{0.8932} \\
& NDCG@10 & 0.5341 & 0.5476 & 0.6475 & 0.5290 & 0.5337 & 0.4951 & 0.5398 & 0.5105 & 0.5425 & 0.5455 & \underline{0.6992} & \textbf{0.7701} \\
\midrule
\multirow{2}{*}{\textbf{ML-100K}} & HR@10 & 0.6585 & 0.6066 & 0.8356 & 0.4889 & 0.4252 & 0.4687 & 0.6882 & 0.4952 & 0.7010 & 0.8897 & \underline{0.9512} & \textbf{0.9958} \\
& NDCG@10 & 0.3781 & 0.3398 & 0.5754 & 0.2721 & 0.2290 & 0.2702 & 0.3913 & 0.2772 & 0.4069 & \underline{0.7950} & 0.7741 & \textbf{0.9427} \\
\midrule
\multirow{2}{*}{\textbf{ML-1M}} & HR@10 & 0.6053 & 0.5897 & 0.8217 & 0.4871 & 0.4230 & 0.4137 & 0.6730 & 0.4533 & 0.6776 & 0.8619 & \underline{0.9012} & \textbf{0.9507} \\
& NDCG@10 & 0.3376 & 0.3325 & 0.5478 & 0.2733 & 0.2285 & 0.2333 & 0.3898 & 0.2475 & 0.3973 & 0.7661 & \underline{0.7890} & \textbf{0.8392} \\
\midrule
\multirow{2}{*}{\textbf{LastFM-2K}} & HR@10 & 0.8440 & 0.7904 & 0.8463 & 0.5934 & 0.4996 & 0.5225 & 0.7833 & 0.5335 & \textbf{0.7975} & 0.6210 & 0.7108 & \underline{0.7935} \\
& NDCG@10 & 0.6161 & 0.6024 & 0.6890 & 0.3963 & 0.3307 & 0.3239 & \underline{0.6822} & 0.3580 & 0.6690 & 0.5923 & 0.6601 & \textbf{0.7151} \\
\bottomrule
\end{tabular}
}
\caption{Performance comparison on four datasets, reported by HR@10 and NDCG@10. CenRec and FedRec represent centralized and federated recommendation methods, respectively. The best FedRec results are bold, and the second ones are underlined.}
\label{tab:exp_main}
\end{table*}

The performance comparison of different methods is summarized in Table~\ref{tab:exp_main}. FedEM outperforms centralized methods on most datasets. As a pFR method, FedEM customizes item embeddings for each client, enabling it to better capture user preferences compared to centralized methods that share the same item embeddings across all clients. FedEM also surpasses other federated methods. Existing pFR methods heuristically design additional personalization mechanisms without tackling the aggregation bottleneck, \textit{i.e.}, the loss of local personalization caused by global aggregation. In contrast, our method effectively and efficiently leverages off-the-shelf personalized information from the local model to bridge this gap. 
Additionally, on relatively dense datasets such as ML-100K and ML-1M, pFR methods, including FedEM, can surpass centralized methods due to the locally sufficient training of personalized models. However, for highly sparse datasets such as LastFM-2K, with a sparsity level of up to $99.07\%$, the federated setting severely limits local data availability, making it difficult to outperform centralized methods. Nevertheless, FedEM still achieves superior performance compared to all pFR methods under such challenging scenarios, demonstrating its general applicability.

\begin{figure}[ht]
\centering
\includegraphics[width=0.98\columnwidth]{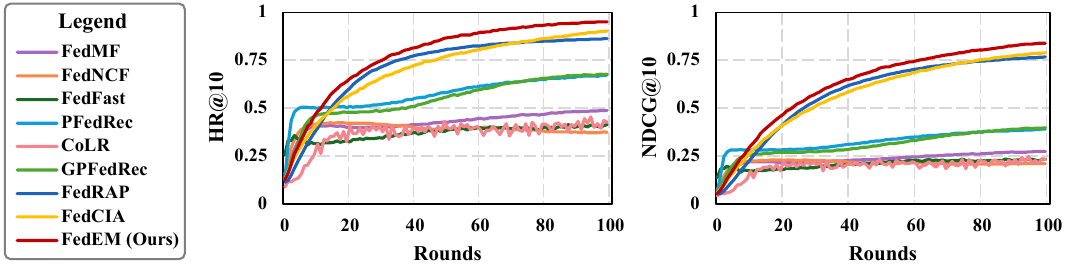} 
\caption{Convergence comparison on the ML-1M dataset.}
\label{pic:convergence_ML-1M}
\end{figure}

We further compare the convergence of different methods, with results on the ML-1M dataset shown in Figure~\ref{pic:convergence_ML-1M}. Compared with other methods, FedEM demonstrates a faster convergence in the early stages and achieves more stable and superior performance in the later stages. This is attributed to the elastic merging scheme, which retains the local information from the previous round, thereby enhancing both performance and convergence stability.

\begin{table}[ht]
    \centering
    \resizebox{0.485\textwidth}{!}{
    \large
    \begin{tabular}{llcccccc}
    \toprule
    \textbf{Datasets} & \textbf{Metrics} & \textbf{\makecell{FedMF\\w/ SR}} & \textbf{\makecell{FedSim\\w/ SR}} & \textbf{\makecell{FedSim\\w/ SM}} & \textbf{\makecell{FedSim\\w/ DM}} & \textbf{\makecell{FedSim\\w/ EM}}\\
    \midrule
    \multirow{2}{*}{\textbf{Filmtrust}} & \normalsize HR@10 & 0.6577 & 0.7206 & 0.7774 & 0.8433 & \textbf{0.8932} \\
     & \normalsize NDCG@10 & 0.5290 & 0.5501 & 0.5959 & 0.7051 & \textbf{0.7701} \\
    \midrule
    \multirow{2}{*}{\textbf{ML-100K}} & \normalsize HR@10 & 0.4889 & 0.6320 & 0.7190 & 0.9427 & \textbf{0.9958} \\
     & \normalsize NDCG@10 & 0.2721 & 0.4457 & 0.5277 & 0.8175 & \textbf{0.9427} \\
    \midrule
    \multirow{2}{*}{\textbf{ML-1M}} & \normalsize HR@10 & 0.4871 & 0.6166 & 0.7219 & 0.9078 & \textbf{0.9507} \\
     & \normalsize NDCG@10 & 0.2733 & 0.4103 & 0.5278 & 0.7839 & \textbf{0.8392} \\
    \midrule
    \multirow{2}{*}{\textbf{LastFm-2K}} & \normalsize HR@10 & 0.5934 & 0.7258 & 0.7455 & 0.7912 & \textbf{0.7935} \\
     & \normalsize NDCG@10 & 0.3963 & 0.6498 & 0.6633 & 0.7131 & \textbf{0.7151} \\
    \bottomrule
    \end{tabular}
    }
    \caption{Ablation study results. The best results are in bold.}
    \label{tab:exp_ab}
\end{table}

\begin{figure}[t]
\centering
\includegraphics[width=0.98\columnwidth]{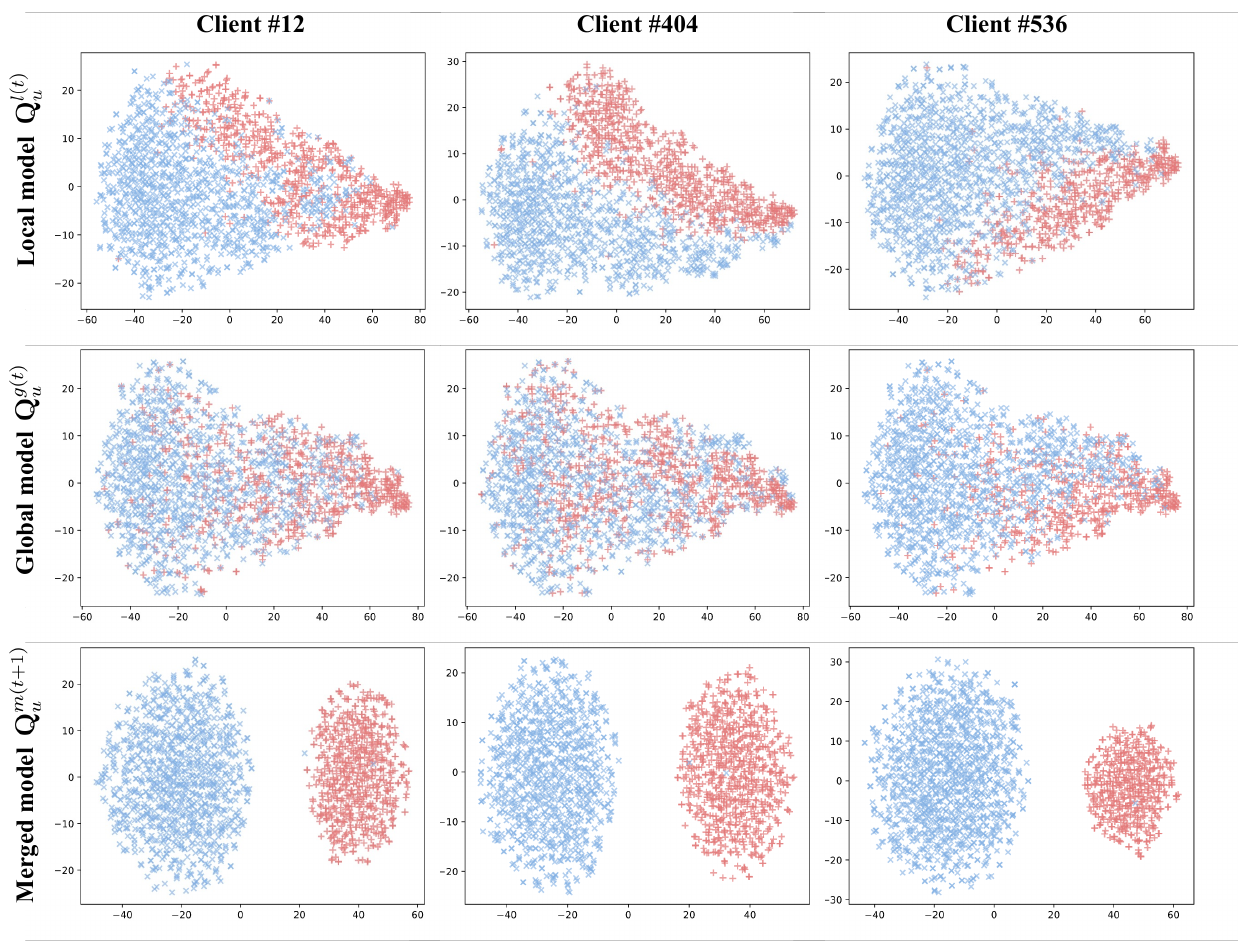} 
\caption{Model visualization of item embeddings on ML-100K. Items interacted with by the user are marked with red $\bm{+}$, while unobserved items are marked with blue $\bm{\times}$.}
\label{pic:vis_ea}
% \vspace{-0.2cm}
\end{figure}

\subsection{Ablation Study}
\textbf{Component analysis.} Based on the backbone FedMF, we apply similarity-based aggregation on the server side and denote this variant as FedSim. Moreover, to validate the rationality and effectiveness of the EM module, we introduce the following variants: (1) Static Replacement (SR): Replace the local model with the global model for local training. (2) Static Merging (SM): Set a fixed scalar $\rho$ shared across all clients to merge the global and local models. (3) Dynamic Merging (DM): Each client adopts a local adapter to generate a personalized scalar $\rho$ for model merging. (4) Elastic Merging (EM): The adapter generates a vector $\bm{\rho}$ to perform fine-grained model merging at the item level for each client.

From the results in Table~\ref{tab:exp_ab}, we observe that incorporating similarity into global aggregation helps mitigate the negative impact of client heterogeneity on local personalization. Focusing on local strategies, the SR scheme adopted by existing methods directly discards the local model, leading to suboptimal performance. By merging the global and local models, SM retains some personalized information and achieves better results. However, SM applies the same merging weight to all clients, neglecting client heterogeneity and diverse personalization needs. DM improves upon SM by introducing the adapter to generate client-specific weights, enabling more tailored merging. Building on DM, EM further leverages the embedding-based nature of recommendation models to perform fine-grained merging at the item level. This design allows for a more balanced fusion of collaborative information from global aggregation and the off-the-shelf personalized information from local training, offering a more efficient and effective solution.

\begin{table*}[htbp]
\centering
\resizebox{0.99\textwidth}{!}{
% \small % 更小字体
\Large
\begin{tabular}{llccccccccccccccc}
\toprule
\multirow{2}{*}{\textbf{Datasets}}& \multirow{2}{*}{\textbf{Metrics}}& \multicolumn{3}{c}{\textbf{FedMF}} & \multicolumn{3}{c}{\textbf{FedNCF}} & \multicolumn{3}{c}{\textbf{PFedRec}} & \multicolumn{3}{c}{\textbf{GPFedRec}} & \multicolumn{3}{c}{\textbf{FedRAP}} \\
\cmidrule(lr){3-5}\cmidrule(lr){6-8}\cmidrule(lr){9-11}\cmidrule(lr){12-14}\cmidrule(lr){15-17}
& & w/o EM & w/ EM & Improv. & w/o EM & w/ EM & Improv. & w/o EM & w/ EM & Improv. & w/o EM & w/ EM & Improv. & w/o EM & w/ EM & Improv. \\
\midrule
\multirow{2}{*}{\textbf{Filmtrust}} & HR@10 & 0.6577 & 0.7804 & $\uparrow 18.66\%$ & 0.6597 & 0.7315 & $\uparrow 10.88\%$ & 0.6756 & 0.9691 & $\uparrow 43.44\%$ & 0.6836 & 0.8723 & $\uparrow 27.60\%$ & 0.8024 & 0.8044 & $\uparrow 0.25\%$ \\
& NDCG@10 & 0.5290 & 0.6426 & $\uparrow 21.47\%$ & 0.5337 & 0.6852 & $\uparrow 28.39\%$ & 0.5398 & 0.8544 & $\uparrow 58.28\%$ & 0.5425 & 0.6706 & $\uparrow 23.61\%$ & 0.5455 & 0.5498 & $\uparrow 0.79\%$ \\
\midrule
\multirow{2}{*}{\textbf{ML-100K}} & HR@10 & 0.4889 & 0.8759 & $\uparrow 79.16\%$ & 0.4252 & 0.7529 & $\uparrow 77.07\%$ & 0.6882 & 0.9905 & $\uparrow 43.93\%$ & 0.7010 & 0.9247 & $\uparrow 31.91\%$ & 0.8897 & 0.9735 & $\uparrow 9.42\%$ \\
& NDCG@10 & 0.2721 & 0.7533 & $\uparrow 176.85\%$ & 0.2290 & 0.7162 & $\uparrow 212.75\%$ & 0.3913 & 0.9238 & $\uparrow 136.08\%$ & 0.4069 & 0.6938 & $\uparrow 70.51\%$ & 0.7950 & 0.8964 & $\uparrow 12.75\%$ \\
\midrule
\multirow{2}{*}{\textbf{ML-1M}} & HR@10 & 0.4871 & 0.8199 & $\uparrow 68.32\%$ & 0.4230 & 0.7606 & $\uparrow 79.81\%$ & 0.6730 & 0.9656 & $\uparrow 43.48\%$ & 0.6776 & 0.8823 & $\uparrow 30.21\%$ & 0.8619 & 0.9538 & $\uparrow 10.66\%$ \\
& NDCG@10 & 0.2733 & 0.6606 & $\uparrow 141.71\%$ & 0.2285 & 0.6979 & $\uparrow 205.43\%$ & 0.3898 & 0.7978 & $\uparrow 104.67\%$ & 0.3973 & 0.6777 & $\uparrow 70.58\%$ & 0.7661 & 0.8648 & $\uparrow 12.88\%$ \\
\midrule
\multirow{2}{*}{\textbf{LastFM-2K}} & HR@10& 0.5934 & 0.7013 & $\uparrow 18.18\%$ & 0.4925 & 0.6564 & $\uparrow 33.28\%$ & 0.7833 & 0.8810 & $\uparrow 12.47\%$ & 0.7975 & 0.8385 & $\uparrow 5.14\%$ & 0.6210 & 0.6320 & $\uparrow 1.77\%$ \\
& NDCG@10 & 0.3963 & 0.6061 & $\uparrow 52.94\%$ & 0.3215 & 0.5445 & $\uparrow 69.36\%$ & 0.6822 & 0.8052 & $\uparrow 18.03\%$ & 0.6690 & 0.7475 & $\uparrow 11.73\%$ & 0.5923 & 0.5869 & -- \\
\bottomrule
\end{tabular}
}
\caption{Performance improvement by integrating the Elastic Merging (EM) module into existing FR/pFR baselines. “Improv.” indicates the performance gain over the original baselines.}
\label{tab:exp_plugin}
\end{table*}

% To better understand the effect of EM module,
\noindent\textbf{Model Visualization.} We incorporate the EM module into FedMF and perform t-SNE visualization of the local, global, and merged models, as shown in Figure~\ref{pic:vis_ea}. The locally trained model contains rich personalized information, with the preferred items being relatively concentrated. Although global aggregation introduces more diverse collaborative information into the model, it also disrupts user preferences. Using such a global model for local training can harm personalization. In contrast, the EM module performs model merging at the item level, effectively leveraging the off-the-shelf personalized information from the local model while selectively incorporating collaborative information from the global model, leading to better performance.

\subsection{Compatibility Study}
Our proposed EM module can be seamlessly integrated as a plug-in into existing FR/pFR methods, aiming to overcome the performance bottleneck caused by global aggregation. Specifically, we take FedMF, FedNCF, PFedRec, GPFedRec, and FedRAP as examples, where elastic merging is performed after the client downloads the global model, instead of following static replacement. As shown in Table~\ref{tab:exp_plugin}, all methods exhibit significant performance improvements after incorporating the EM module, validating both the limitations of the static replacement and the effectiveness of elastic merging. The EM module is lightweight and easily integrable, incurring little to no overhead on local clients and demonstrating good compatibility with existing models.

\subsection{Hyperparameter Analysis}

\textbf{Adapter architecture.} In our experiments, we set the embedding dimension to $d=16$, and by default, the adapter is implemented as an MLP with three hidden layers, following the structure $[32, 16, 8, 1]$. Additionally, we explore MLPs with varying numbers of hidden layers, and the corresponding results are illustrated in Figure~\ref{pic:param_structure}. When the number of MLP layers $L$ is small, the adapter struggles to effectively balance the local and global models, leading to suboptimal results. As $L$ increases, the model performance does not improve significantly, but the adapter becomes more complex. The default structure achieves a good trade-off between performance and architectural simplicity.

\noindent \textbf{Similarity Coefficient.} During global aggregation, we find that setting the similarity coefficient $\alpha$ around 1 can effectively alleviate the negative impact of heterogeneity discussed in Lemma~\ref{lemma:misalignment} and improve overall model performance.

\begin{figure}[ht]
\centering
\includegraphics[width=0.95\columnwidth]{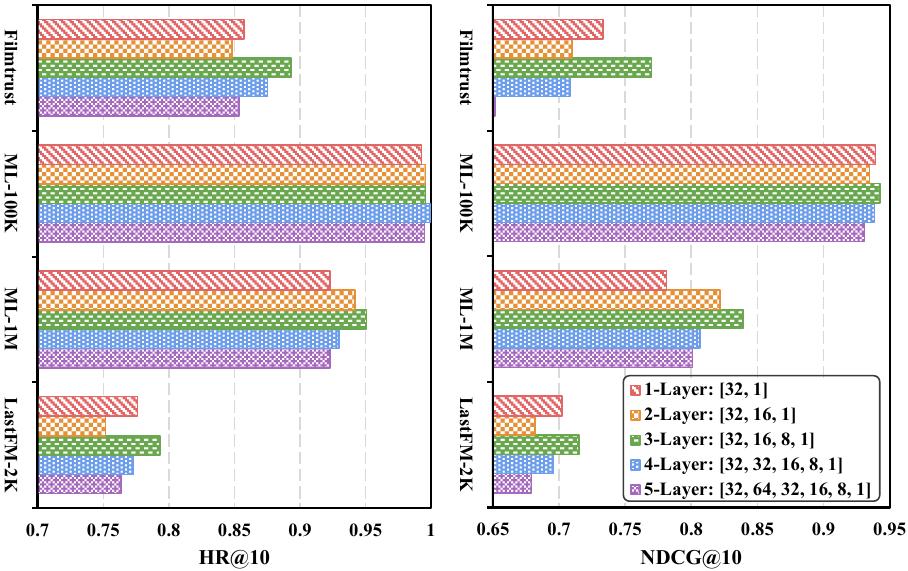} 
\caption{Performance under different adapter architectures.}
\label{pic:param_structure}
\end{figure}

\subsection{Privacy Protection}
We evaluate the performance of privacy-enhanced FedEM by incorporating the local differential privacy (LDP) strategy. As shown in Table~\ref{tab:exp_ldp}, the performance on all datasets slightly degrades with LDP, but remains within an acceptable range and still outperforms other baselines, demonstrating the robustness of our method.

\begin{table}[ht]
    \centering
    \resizebox{0.48\textwidth}{!}{
    \Large
    \begin{tabular}{lcccccccc}
    \toprule
    \multirow{2}{*}{\textbf{Datasets}} & \multicolumn{2}{c}{\textbf{Filmtrust}} & \multicolumn{2}{c}{\textbf{ML-100K}} & \multicolumn{2}{c}{\textbf{ML-1M}} & \multicolumn{2}{c}{\textbf{LastFM-2K}} \\ \cmidrule(lr){2-3}\cmidrule(lr){4-5}\cmidrule(lr){6-7}\cmidrule(lr){8-9}
    & \normalsize HR@10 & \normalsize NDCG@10 & \normalsize HR@10 & \normalsize NDCG@10 & \normalsize HR@10 & \normalsize NDCG@10 & \normalsize HR@10 & \normalsize NDCG@10 \\
    \midrule
    \textbf{w/o LDP} & 0.8932 & 0.7701 & 0.9958 & 0.9427 & 0.9507 & 0.8392 & 0.7935 & 0.7151 \\
    \textbf{w/ LDP} & 0.8802 & 0.7353 & 0.9936 & 0.9399 & 0.9392 & 0.8193 & 0.7707 & 0.6930 \\
    \textbf{Degradation} & $\downarrow 1.46\%$ & $\downarrow 4.52\%$ & $\downarrow 0.22\%$ & $\downarrow 0.30\%$ & $\downarrow 1.21\%$ & $\downarrow 2.37\%$ & $\downarrow 2.87\%$ & $\downarrow 3.09\%$ \\
    \bottomrule
    \end{tabular}
    }
    \caption{Performance of applying LDP to our method.}
    \label{tab:exp_ldp}
\end{table}

\noindent Extensive experimental results are provided in \textit{\textbf{Appendix E}}.

\section{Conclusion}
In this work, we experimentally and theoretically identify a performance bottleneck in FR caused by global aggregation. Unlike existing methods that heuristically design personalization mechanisms, we address this bottleneck directly with a simple yet effective method called FedEM. Grounded in model merging theory, FedEM exploits the off-the-shelf local models to compensate for the aggregated global model. It elastically strikes a balance between global collaboration and local personalization, achieving excellent performance and compatibility compared to other SOTA methods.

% \nobibliography*
\bibliography{aaai2026}

\appendix
\section{Appendix Overview}
We provide the following supplementary materials to further support and elaborate on the main content of our paper:

\begin{enumerate}[label=\Alph*.]
    \item \textbf{Detailed Theoretical Analysis}: We present formal assumptions and detailed proofs for our theoretical results, including the Performance Bottleneck of Global Aggregation and the Compensation Effect of Elastic Merging.
	\item \textbf{Algorithms}: A comprehensive description of the FedEM algorithm is provided, including the full procedural flow and pseudocode.
	\item \textbf{Detailed Discussions}: We further discuss the complexity and privacy of the proposed method.
	\item \textbf{Detailed Experimental Settings}: We elaborate on the datasets, evaluation protocols, compared baselines, and implementation details used in our experiments.
	\item \textbf{Extensive Experimental Results}: We include additional experimental results covering performance comparisons, convergence behavior, parameter visualizations, hyperparameter analysis, and privacy-related evaluations.
\end{enumerate}

\section{A~~~~Detailed Theoretical Analysis}

To facilitate theoretical analysis, we begin by stating a set of standard assumptions regarding the local objective functions. Specifically, we assume that each client’s local loss function $\mathcal{L}_u(\mathbf{Q})$ is continuously differentiable, $L$-smooth, and $\mu$-strongly convex with respect to the model parameter $\mathbf{Q}$. In particular, the binary cross-entropy (BCE) loss we adopt is smooth and convex~\cite{boyd2004convex, shalev2014understanding}, and the addition of $\ell_2$ regularization renders the overall objective approximately strongly convex~\cite{bottou2018optimization}. Additionally, such assumptions are also standard in theoretical analyses of federated learning~\cite{mcmahan2017communication, t2020personalized}.

\begin{assumption}[$L$-smoothness]
The local loss function $\mathcal{L}_u$ of client $u$ is assumed to be \emph{$L$-smooth}, i.e., its gradient is $L$-Lipschitz continuous. That is, for any $\mathbf{Q}_1, \mathbf{Q}_2$, it holds that:
\begin{equation}
    \|\nabla \mathcal{L}_u(\mathbf{Q}_1) - \nabla \mathcal{L}_u(\mathbf{Q}_2)\| \leq L \|\mathbf{Q}_1 - \mathbf{Q}_2\|.
\end{equation}
An equivalent characterization of $L$-smoothness is the following upper bound on the function value:
\begin{equation}
\begin{split}
    \mathcal{L}_u(\mathbf{Q}_2) \leq \mathcal{L}_u(\mathbf{Q}_1) & + \langle \nabla \mathcal{L}_u(\mathbf{Q}_1), \mathbf{Q}_2 - \mathbf{Q}_1 \rangle \\ & + \frac{L}{2} \|\mathbf{Q}_2 - \mathbf{Q}_1\|^2.
\end{split}
\end{equation}
Here, $L>0$ is the smoothness constant, which quantifies the maximum rate of change in the gradient. Intuitively, $L$-smoothness implies that the loss surface is locally well-behaved, enabling stable gradient-based optimization.
\end{assumption}

\begin{assumption}[$\mu$-strong convexity]
The local loss function $\mathcal{L}_u$ of client $u$ is assumed to be \emph{$\mu$-strongly convex}, i.e., it satisfies the following inequality for any $\mathbf{Q}_1, \mathbf{Q}_2$:
\begin{equation}
\begin{split}
    \mathcal{L}_u(\mathbf{Q}_2) \geq \mathcal{L}_u(\mathbf{Q}_1) & + \left\langle \nabla \mathcal{L}_u(\mathbf{Q}_1), \mathbf{Q}_2 - \mathbf{Q}_1 \right\rangle \\
    & + \frac{\mu}{2} \left\| \mathbf{Q}_2 - \mathbf{Q}_1 \right\|^2.
\end{split}
\end{equation}
Here, $\mu > 0$ is the strong convexity constant, which quantifies the curvature of the loss surface. Intuitively, $\mu$-strong convexity ensures that $\mathcal{L}_u$ has a unique minimizer, and that gradient-based updates will converge rapidly and stably to this minimum.
\end{assumption}

Based on the aforementioned assumptions, we first examine the limitations of global aggregation in federated recommendation. We then provide a theoretical analysis of how the proposed elastic merging scheme alleviates this issue.
\begin{lemma}[Performance Bottleneck of Global Aggregation]\label{lemma:aggregation_bottleneck_}
    Let $\mathbf{Q}_u^*$ denote the local optimal model. Following the static replacement scheme, i.e., using the global model $\mathbf{Q}^{g(t)}$ for local training at round $t+1$, may lead to:
    \begin{equation}
        \langle \nabla \mathcal{L}_u(\mathbf{Q}^{g(t)}), \mathbf{Q}^{g(t)} - \mathbf{Q}_u^* \rangle \leq 0,
    \end{equation}
    which suggests that the global aggregated model can cause the training optimization to deviate from the client-specific optimum, thereby undermining local personalization. This issue becomes more pronounced under the inherent high client heterogeneity in FR, leading to a performance bottleneck.
\end{lemma}
\begin{proof}
    Most existing FR methods employ weighted aggregation on the server side to share collaborative information across clients. Taking FCF~\cite{ammad2019federated} for example, that is,
    \begin{equation}
    \begin{split}
        \mathbf{Q}^{g(t)} & = \frac{1}{n}\sum_{u=1}^{n}\mathbf{Q}_u^{l(t)}=\frac{1}{n} \mathbf{Q}_u^{l(t)} + \frac{n - 1}{n} \mathbf{Q}_r^{g(t)}, \\ \quad & \text{where} \quad \mathbf{Q}_r^{g(t)} = \frac{1}{n - 1} \sum_{v \ne u}^n \mathbf{Q}_v^{l(t)}.
    \end{split}
    \end{equation}
    Here, $\mathbf{Q}_r^{g(t)}$ denotes the aggregated model over the remaining clients except ego Client $u$. However, FR tasks often involve tens of thousands or even more clients, so we have,
    \begin{equation}
    \lim_{n\to\infty}\frac{1}{n}=0,\lim_{n\to\infty}\frac{n-1}{n}=1.
    \end{equation}
    Thus, when $n$ is large, $\mathbf{Q}^{g(t)} \approx \mathbf{Q}_r^{g(t)}$, that is, the global model contains almost no information from the own data $\mathcal{D}_u$ of client $u$. 
    Since $\mathcal{L}_u$ is $L$-smooth, we have:
    \begin{equation}\label{eq:grad_bound}
        \| \nabla \mathcal{L}_u(\mathbf{Q}^{g(t)}) - \nabla \mathcal{L}_u(\mathbf{Q}_r^{g(t)}) \| \leq L \| \mathbf{Q}^{g(t)} - \mathbf{Q}_r^{g(t)} \|\approx0.
    \end{equation}
    Thus, we have $\nabla \mathcal{L}_u(\mathbf{Q}^{g(t)}) \approx \nabla \mathcal{L}_u(\mathbf{Q}_r^{g(t)})$ when $n$ is large. The gradient direction of the global model is almost entirely determined by the non-local model $\mathbf{Q}_r^{g(t)}$.
    
    Based on the above derivation, in the following analysis, we use $\mathbf{Q}_r^{g(t)}$ to represent $\mathbf{Q}^{g(t)}$. Denote the client-specific optimum by $\mathbf{Q}_u^* = \arg\min_{\mathbf{Q}} \mathcal{L}_u(\mathbf{Q})$. Then, due to the $\mu$-strong convexity of $\mathcal{L}_u$, for any local model $\mathbf{Q}_u^{l(t)}$ trained on the data distribution $\mathcal{D}_u$, it holds that:
    \begin{equation}
    \label{eq:mu-guarantee}
    \begin{split}
        & \left\langle \nabla  \mathcal{L}_u(\mathbf{Q}_u^{l(t)}), \mathbf{Q}_u^{l(t)} - \mathbf{Q}_u^* \right\rangle \ge \mathcal{L}_u(\mathbf{Q}_u^{l(t)})-\mathcal{L}_u(\mathbf{Q}_u^*)\\ & +\frac{\mu}{2} \|\mathbf{Q}_u^{l(t)} - \mathbf{Q}_u^*\|^2 = \mathcal{L}_u(\mathbf{Q}_u^{l(t)})+\frac{\mu}{2} \|\mathbf{Q}_u^{l(t)}-\mathbf{Q}_u^*\|^2 \\ & \ge \frac{\mu}{2} \|\mathbf{Q}_u^{l(t)}-\mathbf{Q}_u^*\|^2 \ge 0,
    \end{split}
    \end{equation}
    which implies that the optimization direction of the trained model $\mathbf{Q}_u^{l(t)}$ is toward its local optimum $\mathbf{Q}_u^*$. 
    
    Under non-IID settings in FR, the optimal solutions  $\mathbf{Q}_u^*$ of local loss functions $\mathcal{L}_u$ differ significantly across clients. That is, for $v \ne u$, we have $\|\mathbf{Q}_v^* - \mathbf{Q}_u^*\| \gg 0$.
    In this case, the aggregated reference model $\mathbf{Q}_r^{g(t)} = \frac{1}{n - 1} \sum_{v \ne u}^n \mathbf{Q}_v^{l(t)}$, naturally deviates from the local optimum $\mathbf{Q}_u^*$, leading to $\|\mathbf{Q}_r^{g(t)} - \mathbf{Q}_u^*\| > 0$. Following prior works on federated learning under heterogeneous settings~\cite{karimireddy2020scaffold, li2020federated}, such deviation is lower bounded by a constant that depends on the level of data heterogeneity. Specifically, it holds that:
    \begin{equation}
        \|\mathbf{Q}_r^{g(t)} - \mathbf{Q}_u^*\| \ge \Omega(\zeta),
    \end{equation}
    where $\zeta$ is a heterogeneity coefficient. A larger $\zeta$ indicates a greater divergence between the global model and the local optimum. Consequently, when $\zeta$ is sufficiently large, the model $\mathbf{Q}_r^{g(t)}$ may lie on the wrong side of the loss landscape with respect to $\mathcal{L}_u$. Specifically, the condition implied in Eq.~(\ref{eq:mu-guarantee}) may no longer hold, that is,
    \begin{equation}
        \langle \nabla \mathcal{L}_u(\mathbf{Q}_r^{g(t)}), \mathbf{Q}_r^{g(t)} - \mathbf{Q}_u^* \rangle \leq 0.
    \end{equation}
    Under the traditional static replacement scheme, the next local gradient descent update is given by:
    \begin{equation}
        \mathbf{Q}_{u}^{l(t+1)} = \mathbf{Q}_r^{g(t)} - \eta \nabla \mathcal{L}_u(\mathbf{Q}_r^{g(t)}),
    \end{equation}
    where $\mathbf{Q}_r^{g(t)}$ is the received global model. We analyze the squared distance between the updated model and the local optimum $\mathbf{Q}_u^*$:
    \begin{equation}
    \begin{split}
        & \|\mathbf{Q}_{u}^{l(t+1)} - \mathbf{Q}_u^*\|^2 = \| \mathbf{Q}_r^{g(t)} - \mathbf{Q}_u^* \|^2\\ & - 2\eta \langle \nabla \mathcal{L}_u(\mathbf{Q}_r^{g(t)}), \mathbf{Q}_r^{g(t)} - \mathbf{Q}_u^* \rangle + \eta^2 \| \nabla \mathcal{L}_u(\mathbf{Q}_r^{g(t)}) \|^2.
    \end{split}
    \end{equation}
    Rearranging terms, we can obtain:
    \begin{equation}
    \begin{split}
        & \|\mathbf{Q}_{u}^{l(t+1)} - \mathbf{Q}_u^*\|^2 - \| \mathbf{Q}_r^{g(t)} - \mathbf{Q}_u^* \|^2 = \\ & - 2\eta \langle \nabla \mathcal{L}_u(\mathbf{Q}_r^{g(t)}), \mathbf{Q}_r^{g(t)} - \mathbf{Q}_u^* \rangle + \eta^2 \| \nabla \mathcal{L}_u(\mathbf{Q}_r^{g(t)}) \|^2 \geq 0.
    \end{split}
    \end{equation}
    Therefore, we have:
    \begin{equation}
        \|\mathbf{Q}_{u}^{l(t+1)} - \mathbf{Q}_u^* \|^2 \geq \| \mathbf{Q}_r^{g(t)} - \mathbf{Q}_u^*\|^2\ge \Omega(\zeta^2),
    \end{equation}
    which indicates that the local model gradually moves away from the local optimum. These analyses demonstrate that the local model may be optimized in a direction opposite to the local objective, which can hinder convergence and degrade personalization. This issue becomes more pronounced as the heterogeneity level $\zeta$ increases, thereby resulting in a performance bottleneck during global aggregation.

    \noindent Hence, the proof is complete.\end{proof}

\begin{lemma}[Compensation Effect of Elastic Merging]\label{lemma:elastic_merging_}
Following the elastic merging scheme, we define the merged model as:
\begin{equation}
    \mathbf{Q}_u^{m(t+1)} = \rho \mathbf{Q}^{g(t)} + (1 - \rho) \mathbf{Q}_u^{l(t)}, \quad \rho \in [0, 1]  .
\end{equation}
The merged model can compensate for the optimization deviation caused by aggregation, guiding the update toward the client-specific optimum, \textit{i.e.}, it satisfies:
\begin{equation}
\begin{split}
    & \langle \nabla \mathcal{L}_u(\mathbf{Q}_u^{m(t+1)}), \mathbf{Q}_u^{m(t+1)} - \mathbf{Q}_u^* \rangle\\ & \geq (1 - \rho) \langle \nabla \mathcal{L}_u(\mathbf{Q}_u^{l(t)}), \mathbf{Q}_u^{l(t)} - \mathbf{Q}_u^* \rangle-\rho C,
\end{split}
\end{equation}
where $C$ denotes the deviation introduced by global aggregation. This deviation can be lower bounded by both the state of the local model $\mathbf{Q}_u^{l(t)}$ and the global-local discrepancy $\mathbf{Q}_u^{\Delta(t)} = \mathbf{Q}^{g(t)} - \mathbf{Q}_u^{l(t)}$. Specifically, the bound is given by $C=\|\nabla \mathcal{L}_u(\mathbf{Q}_u^{l(t)})\|(\|\mathbf{Q}_u^{l(t)} - \mathbf{Q}_u^*\|+\|\mathbf{Q}_u^{\Delta(t)}\|)+ L \|\mathbf{Q}_u^{\Delta(t)}\|(\|\mathbf{Q}_u^{l(t)} - \mathbf{Q}_u^*\|+\rho\|\mathbf{Q}_u^{\Delta(t)}\|)$. The first term on the right-hand side of the inequality, $A = \langle \nabla \mathcal{L}_u(\mathbf{Q}_u^{l(t)}),\, \mathbf{Q}_u^{l(t)} - \mathbf{Q}_u^* \rangle>0$, reflects the preservation of client-specific optimization information. Thus, we can adaptively control the merging weight $0\leq\rho < \frac{A}{A + C}$ to facilitate local personalized optimization during global collaborative training.
\end{lemma}
\begin{proof}
    Based on the definition $\mathbf{Q}_u^{m(t+1)} = \rho \mathbf{Q}^{g(t)} + (1 - \rho) \mathbf{Q}_u^{l(t)}$, we define the global-local discrepancy as $\mathbf{Q}_u^{\Delta(t)} = \mathbf{Q}^{g(t)} - \mathbf{Q}_u^{l(t)}$, such that the merged model can be rewritten as $\mathbf{Q}_u^{m(t+1)} = \mathbf{Q}_u^{l(t)} + \rho \mathbf{Q}_u^{\Delta(t)}$. By $L$-smoothness, we have:
    \begin{equation}
        \nabla \mathcal{L}_u(\mathbf{Q}_u^{m(t+1)}) = \nabla \mathcal{L}_u(\mathbf{Q}_u^{l(t)}) + \mathbf{R},
    \end{equation}
    where $\|\mathbf{R}\| \leq L \|\mathbf{Q}_u^{m(t+1)} - \mathbf{Q}_u^{l(t)}\| = L \rho \|\mathbf{Q}_u^{\Delta(t)}\|$.
    We now analyze the optimization direction of the merged model:
    \begin{equation}
    \begin{split}
        & \langle \nabla \mathcal{L}_u(\mathbf{Q}_u^{m(t+1)}), \mathbf{Q}_u^{m(t+1)} - \mathbf{Q}_u^* \rangle\\ & = \langle \nabla \mathcal{L}_u(\mathbf{Q}_u^{l(t)}) + \mathbf{R},\, \mathbf{Q}_u^{m(t+1)} - \mathbf{Q}_u^* \rangle \\ &= \langle \nabla \mathcal{L}_u(\mathbf{Q}_u^{l(t)}),\, \mathbf{Q}_u^{m(t+1)} - \mathbf{Q}_u^* \rangle + \langle \mathbf{R},\, \mathbf{Q}_u^{m(t+1)} - \mathbf{Q}_u^* \rangle\\ & = \langle \nabla \mathcal{L}_u(\mathbf{Q}_u^{l(t)}),\, (1 - \rho)(\mathbf{Q}_u^{l(t)} - \mathbf{Q}_u^*) + \rho (\mathbf{Q}^{g(t)} - \mathbf{Q}_u^*) \rangle\\ & \quad\ + \langle \mathbf{R},\, \mathbf{Q}_u^{m(t+1)} - \mathbf{Q}_u^* \rangle\\ &= (1 - \rho) \langle \nabla \mathcal{L}_u(\mathbf{Q}_u^{l(t)}),\, \mathbf{Q}_u^{l(t)} - \mathbf{Q}_u^* \rangle +\\ & \quad\ (\rho \langle \nabla \mathcal{L}_u(\mathbf{Q}_u^{l(t)}),\, \mathbf{Q}^{g(t)} - \mathbf{Q}_u^* \rangle + \langle \mathbf{R},\, \mathbf{Q}_u^{m(t+1)} - \mathbf{Q}_u^* \rangle).
    \end{split}
    \end{equation}
    Moreover, the lower bound can be formulated as:
    \begin{equation}\label{eq:adapt_lowerbound}
    \begin{split}
        & \langle \nabla \mathcal{L}_u(\mathbf{Q}_u^{m(t+1)}), \mathbf{Q}_u^{m(t+1)} - \mathbf{Q}_u^* \rangle\\ & \geq (1 - \rho) \langle \nabla \mathcal{L}_u(\mathbf{Q}_u^{l(t)}),\, \mathbf{Q}_u^{l(t)} - \mathbf{Q}_u^* \rangle - \\ &\quad\ (\rho \|\langle \nabla \mathcal{L}_u(\mathbf{Q}_u^{l(t)}),\, \mathbf{Q}^{g(t)} - \mathbf{Q}_u^* \rangle\| + \|\langle \mathbf{R},\, \mathbf{Q}_u^{m(t+1)} - \mathbf{Q}_u^* \rangle\|).
    \end{split}
    \end{equation}
    By applying the Cauchy–Schwarz inequality and the triangle inequality, we obtain the following upper bound:
    \begin{equation}
    \begin{split}
        &\| \langle \nabla \mathcal{L}_u(\mathbf{Q}_u^{l(t)}),\, \mathbf{Q}^{g(t)} - \mathbf{Q}_u^* \rangle \|\\ & \le \|\nabla \mathcal{L}_u(\mathbf{Q}_u^{l(t)})\| \cdot \|\mathbf{Q}^{g(t)} - \mathbf{Q}_u^*\|\\&=\|\nabla \mathcal{L}_u(\mathbf{Q}_u^{l(t)})\| \cdot \|\mathbf{Q}_u^{l(t)} - \mathbf{Q}_u^*+\mathbf{Q}_u^{\Delta(t)}\|\\& \le \|\nabla \mathcal{L}_u(\mathbf{Q}_u^{l(t)})\| \cdot (\|\mathbf{Q}_u^{l(t)} - \mathbf{Q}_u^*\|+\|\mathbf{Q}_u^{\Delta(t)}\|)
    \end{split}
    \end{equation}
    Similarly, we have:
    \begin{equation}
    \begin{split}
        &\|\langle \mathbf{R},\, \mathbf{Q}_u^{m(t+1)} - \mathbf{Q}_u^* \rangle\| \le \|\mathbf{R}\| \cdot \|\mathbf{Q}_u^{m(t+1)} - \mathbf{Q}_u^*\|\\&=\|\mathbf{R}\| \cdot \|\mathbf{Q}_u^{l(t)} - \mathbf{Q}_u^*+\rho\mathbf{Q}_u^{\Delta(t)}\|\\& \le L \rho \|\mathbf{Q}_u^{\Delta(t)}\| \cdot (\|\mathbf{Q}_u^{l(t)} - \mathbf{Q}_u^*\|+\rho\|\mathbf{Q}_u^{\Delta(t)}\|)
    \end{split}
    \end{equation}
    Thus, Eq.~(\ref{eq:adapt_lowerbound}) can be reformulated as:
    \begin{equation}\label{eq:adapt_lowerbound_}
    \begin{split}
        & \langle \nabla \mathcal{L}_u(\mathbf{Q}_u^{m(t+1)}), \mathbf{Q}_u^{m(t+1)} - \mathbf{Q}_u^* \rangle\\ & \geq (1 - \rho) \langle \nabla \mathcal{L}_u(\mathbf{Q}_u^{l(t)}),\, \mathbf{Q}_u^{l(t)} - \mathbf{Q}_u^* \rangle \\ &\quad\ - \rho\|\nabla \mathcal{L}_u(\mathbf{Q}_u^{l(t)})\| \cdot (\|\mathbf{Q}_u^{l(t)} - \mathbf{Q}_u^*\|+\|\mathbf{Q}_u^{\Delta(t)}\|)\\
        &\quad\ - L \rho \|\mathbf{Q}_u^{\Delta(t)}\| \cdot (\|\mathbf{Q}_u^{l(t)} - \mathbf{Q}_u^*\|+\rho\|\mathbf{Q}_u^{\Delta(t)}\|).
    \end{split}
    \end{equation}
    
    For simplicity, we denote $A = \langle \nabla \mathcal{L}_u(\mathbf{Q}_u^{l(t)}), \mathbf{Q}_u^{l(t)} - \mathbf{Q}_u^* \rangle$ as the optimization information of the local model, \textit{i.e.}, its original update toward the client-specific optimum. Besides, $C=\|\nabla \mathcal{L}_u(\mathbf{Q}_u^{l(t)})\|(\|\mathbf{Q}_u^{l(t)} - \mathbf{Q}_u^*\|+\|\mathbf{Q}_u^{\Delta(t)}\|)+ L \|\mathbf{Q}_u^{\Delta(t)}\|(\|\mathbf{Q}_u^{l(t)} - \mathbf{Q}_u^*\|+\rho\|\mathbf{Q}_u^{\Delta(t)}\|)$ denotes the optimization deviation introduced by global aggregation. 
    To ensure that the merged model avoids the performance bottleneck analyzed in Lemma~\ref{lemma:aggregation_bottleneck_}, that is, its optimization direction should be aligned with the local objective, we require the left-hand side of Eq.~(\ref{eq:adapt_lowerbound_}) to be non-negative. Thus, we need to adaptively control the merging weight $0\leq\rho < \frac{A}{A + C}$ to compensate for the deficiency of the aggregation bottleneck.
    
    Building on the above analysis, it is necessary to take into account the deviation term $C$, \textit{i.e.}, the state of the local model $\mathbf{Q}_u^{l(t)}$ and the global-local discrepancy $\mathbf{Q}_u^{\Delta(t)}$, to determine a tailored merging weight $\rho$. Intuitively, we need to measure the local model’s need for the collaborative information encoded in $\mathbf{Q}_u^{\Delta(t)}$, so as to guide the merged model toward the local optimum and preserve personalization.
    
    \noindent Hence, the proof is complete.\end{proof}

\section{B~~~~Algorithms}
We present the FedEM algorithm in detail in Algorithm~\ref{alg:FedEM}. At the beginning of each round, client $u$ merges the downloaded global model $\mathbf{Q}_u^{g(t-1)}$ with the off-the-shelf local model $\mathbf{Q}_u^{l(t-1)}$ to derive a merged model $\mathbf{Q}_u^{m(t)}$, while updating the adapter $\Theta_u$. It adopts an elastic merging scheme instead of the static replacement scheme used in existing FR methods. Then, client $u$ updates $\mathbf{Q}_u^{m(t)}$ and other local parameters, \textit{e.g.}, $\mathbf{p}_u^{(t-1)}$ using its own data $\mathcal{D}_u$, and uploads the updated model $\mathbf{Q}_u^{l(t)}$ to the server. The server aggregates $\{\mathbf{Q}_1^{l(t)}, \cdots, \mathbf{Q}_n^{l(t)}\}$ based on similarity to generate a personalized global model $\mathbf{Q}_u^{g(t)}$ for each client. Initially, $\mathbf{Q}_u^{g(0)}$ and $\mathbf{Q}_u^{l(0)}$ are randomly initialized. After $T$ rounds, each client obtains a personalized model.
\begin{algorithm}[tb]
\small
\caption{FedEM}
\label{alg:FedEM}
\raggedright
\textbf{Input:} selected clients $\mathcal{U}_s$; global rounds \(T\); local epochs \(E\); batch size \(B\); learning rates \(\eta,\beta\); tuning coefficients \(\alpha\)\\
\textbf{Output}: personalized model $\mathbf{p}_u^{(T)}$ ,$\mathbf{Q}_u^{l(T)}$, and adapter $\Theta_u$ for each client $u$\\
\textbf{GlobalProcedure:} \\
% \vspace{-0.4cm}
\begin{algorithmic}[1] %[1] enables line numbers
\FOR{each round $t=1,2,\cdots, T$}
    \FOR{each client $u\in \mathcal{U}_s$ \textbf{in parallel}}
        \STATE Downloads global model $\mathbf{Q}_u^{g(t-1)}$ from the server;
        \STATE $\mathbf{Q}_u^{l(t)}\leftarrow$ ClientUpdate($\mathbf{Q}_u^{g(t-1)}$, $\textit{u}$);
        \STATE Uploads local model $\mathbf{Q}_u^{l(t)}$ to the server;
    \ENDFOR
    \STATE Generates $\mathbf{Q}_u^{g(t)}=\sum_{v=1}^n w_{uv}\mathbf{Q}_v^{l(t)}$ based on Eq.(13);\\  \hspace{13em} $\triangleright$ Global Aggregation
\ENDFOR
\end{algorithmic}

{\bf ClientUpdate}$(\mathbf{Q}_u^{g(t-1)},~$\textit{u}$)$: \\
% \vspace{-0.4cm}
\begin{algorithmic}[1]
\FOR{each local epoch $r=1,2,\cdots, E$}
    \FOR{each batch $b=1,2,\cdots,B$ in $\mathcal{D}_u$}
        \STATE Updates $\Theta_u$ with Eq.(11);
        \STATE Yields $\bm{\rho}$ based on $\Theta_u$ via Eq.(9);
        \STATE Obtains $\mathbf{Q}_u^{m(t)} = (\bm{1}-\bm{\rho})\mathbf{Q}_u^{l(t-1)} + \bm{\rho} \odot \mathbf{Q}_u^{g(t-1)}$;
    \ENDFOR
    \hspace{11.5em} $\triangleright$ Elastic Merging
    \FOR{each batch $b=1,2,\cdots,B$ in $\mathcal{D}_u$}
        \STATE Updates $\mathbf{p}_u^{(t-1)}$ and $\mathbf{Q}_u^{m(t)}$ with Eq.(12);
    \ENDFOR
    \hspace{12em} $\triangleright$ Local Training
\ENDFOR
\STATE Obtains updated $\Theta_u$, $\mathbf{p}_u^{(t)}$ and $\mathbf{Q}_u^{l(t)}$ after $E$ epochs;
\RETURN $\mathbf{Q}_u^{l(t)}$
\end{algorithmic}
\end{algorithm}

\section{C~~~~Detailed Discussions}
\subsection{Complexity Analysis}
Given the limited local resources in federated recommendation scenarios, we primarily focus on analyzing the local complexity. Suppose there are $n$ clients and $m$ items, with an embedding dimension of $d$. The computational complexity of the local backbone model is $\mathcal{O}((m + 1)d)$. Our proposed EM module is implemented as an $L$-layer MLP with a complexity of $\mathcal{O}(Ld^2)$. Since in practice $m \gg d > L$, the EM module incurs negligible overhead on local devices, making it well-suited for on-device recommendation.

Considering other SOTA methods, PFedRec introduces an $L$-layer personalized score function with a similar complexity of $\mathcal{O}(Ld^2)$, which is lightweight but delivers inferior performance compared to our method. Other high-performing approaches, such as FedRAP, require training and storing an additional model locally, resulting in a complexity overhead of $\mathcal{O}(md)$. FedCIA incurs an even higher cost of $\mathcal{O}(m^2d)$ due to the computation of a local item similarity matrix. These methods are therefore unsuitable for resource-constrained on-device scenarios. In contrast, our method strikes a favorable balance between performance and computational complexity.

\subsection{Privacy Analysis}
\textbf{Inherent Federated Framework.} FedEM strictly follows the setting of FR that no raw data is shared with the third parties, safeguarding user privacy and data security~\cite{sun2022survey,yin2024device}. 

\noindent \textbf{Proposed Elastic Merging.} In the traditional federated framework, each client updates the downloaded global model $\mathbf{Q}_u^{g(t-1)}$ to obtain its local model $\mathbf{Q}_u^{l(t)}$, which is then uploaded to the server. This direct update makes it feasible for the server to infer client-specific data distributions via model differences. Formally, $\mathcal{D}_u \leftarrow \text{Attack}(\mathbf{Q}_u^{l(t)} - \mathbf{Q}_u^{g(t-1)})$, where $\text{Attack}(\cdot)$ denotes any reconstruction or inference attack algorithm, such as Gradients Leak~\cite{chai2020secure}.
By contrast, FedEM introduces local merging before training. Each client first interpolates a personalized initialization $\mathbf{Q}_u^{m(t)} = \bm{\rho} \mathbf{Q}_u^{g(t-1)} + (\mathbf{1} - \bm{\rho}) \mathbf{Q}_u^{l(t-1)}$,
and then performs local training on $\mathbf{Q}_u^{m(t)}$ to obtain $\mathbf{Q}_u^{l(t)}$. As a result, the attack surface shifts to $\mathcal{D}_u \leftarrow \text{Attack}(\mathbf{Q}_u^{l(t)} - \mathbf{Q}_u^{m(t)})$,
where the server must know the internal merging weights $\bm{\rho}$ to reconstruct $\mathbf{Q}_u^{m(t)}$. Since $\bm{\rho}$ is generated locally by the client and not exposed to the server, FedEM effectively obfuscates the model update trajectory, increasing the difficulty of inference attacks and enhancing privacy protection.

\noindent \textbf{Privacy-Enhanced FedEM.} Our proposed elastic merging is model agnostic, which can be smoothly integrated with other privacy-enhanced FR models,~\textit{e.g.}, FMF-LDP~\cite{minto2021stronger}. Besides, advanced privacy protection strategies, such as differential privacy~\cite{qi2024towards} and homomorphic encryption~\cite{zhang2020batchcrypt}, are also applicable to our method, further enhancing user privacy. Following previous works~\cite{dwork2006calibrating,choi2018guaranteeing}, we introduce the differential privacy strategy into our method by incorporating some zero-mean Laplacian noises to the uploaded models, that is,
\begin{equation}
\label{eq:dp}
\mathbf{Q}_u^{l*}=\mathbf{Q}_u^{l}+\text{Laplacian}(0,\delta),
\end{equation}
\noindent where $\delta=\mathcal{S}_u/\varepsilon$ is the noise strength~\cite{dwork2006calibrating}. $\mathcal{S}_u$ represents the global sensitivity of client $u$, whose upper bound is derived in Lemma~\ref{thm:ldp}. A higher value of $\delta$ corresponds to stronger privacy guarantees.
\begin{lemma}[Upper Bound of Global Sensitivity $\mathcal{S}_u$]\label{thm:ldp}
Consider two global models, $\mathbf{Q}^g$ and $\mathbf{Q}^{g\prime}$, trained on datasets differing only by the data of client $u$, \textit{i.e.}, $\mathcal{D}_u$ and $\mathcal{D}_u^{\prime}$. Then the following inequality holds:
\begin{equation}
\begin{split}
    \mathcal{S}_u & = \| \mathbf{Q}^g - \mathbf{Q}^{g\prime} \| = \| p_u \eta (\nabla \mathbf{Q}_u^{l} - \nabla \mathbf{Q}_u^{l\prime}) \| \\ & \le p_u \eta (\| \nabla \mathbf{Q}_u^{l} \| + \| \nabla \mathbf{Q}_u^{l\prime} \|) \le 2 p_u \eta Z,
\end{split}
\end{equation}
where $p_u$ is the aggregation weight of client $u$’s local model, and $\eta$ is the local learning rate for updating embeddings. The gradients $\nabla \mathbf{Q}_u^{l}$ and $\nabla \mathbf{Q}_u^{l\prime}$ are computed on datasets $\mathcal{D}_u$ and $\mathcal{D}_u^{\prime}$, respectively. They can be bounded via gradient clipping by a predefined constant $Z$~\cite{wei2020federated}.
\end{lemma}

\section{D~~~~Detailed Experimental Settings}
\subsection{Datasets} 
We experiment on four benchmark datasets with varying client scales and data sparsity to comprehensively evaluate our proposed method. \textbf{Filmtrust}~\cite{filmtrust_2013}, \textbf{Movielens-100K (ML-100K)}~\cite{movielens_2015} and \textbf{Movielens-1M (ML-1M)}~\cite{movielens_2015} are for movie recommendation with user-movie ratings, \textbf{LastFM-2K}~\cite{cantador2011second} is for music recommendation with user-artist listening counts. Here, we convert ratings/counts greater than 0 to 1 to produce implicit data. In this work, we filter out all the users with less than 10 interactions from the above four datasets. Besides, each user is treated as an independent client, and each client’s data inherently exhibits great heterogeneity. The characteristics of the four datasets are summarized in Table~\ref{tab:dataset_statistic}.

\begin{table}[thb]\centering
    \resizebox{0.46\textwidth}{!}{
    \large
    \begin{tabular}{*{5}{c}}
        \toprule
       Datasets & \#User/Client & \#Item & \#Interaction & Sparsity \\
        \midrule
        Filmtrust & 1,002 & 2,042 & 33,372 & 98.37\% \\
        ML-100K & 943 & 1,682 & 100,000 & 93.70\% \\
        ML-1M & 6,040 & 3,706 & 1,000,209 & 95.53\% \\
        LastFm-2K & 1,600 & 12,454 & 185,650 & 99.07\% \\
        \bottomrule
    \end{tabular}
    }
    \caption{Statistics of the experimental datasets.}
    \label{tab:dataset_statistic}
\end{table}

\subsection{Evaluation Protocols} 
We follow the popular \textit{leave-one-out} evaluation \cite{bayer2017generic, he2017neural}, that is, hold out the latest interaction as the test set and regard the remaining data as train set. In addition, we take the last interaction of train set as the validation set for tuning hyper-parameters. To alleviate the high computational cost of ranking all items for each user, we follow the common practice~\cite{he2017neural, zhang2023dual}. Specifically, for each user, we sample 99 negatives not seen during training and rank the test instance among the 100 items. Besides, the model performance is reported by \textit{Hit Ratio} (HR@$K$) and \textit{Normalized Discounted Cumulative Gain} (NDCG@$K$) metrics \cite{he2015trirank}, indicating whether the test item is recommended and its rank in the top $K$ recommended items, respectively. In this work, we set $K=10$, and report the results as the average of 5 repeated experiments.

\begin{table*}[ht]
\centering
\resizebox{\textwidth}{!}{
\Large
\begin{tabular}{llcccccccccccc}
\toprule
\multirow{2}{*}{\textbf{Datasets}} & \multirow{2}{*}{\textbf{Methods}}& \multicolumn{3}{c}{\textbf{CenRec}} & \multicolumn{8}{c}{\textbf{FedRec}} & \multicolumn{1}{c}{\textbf{Ours}} \\
\cmidrule(lr){3-5}\cmidrule(lr){6-13}\cmidrule(lr){14-14}& & \textbf{MF} & \textbf{NCF} & \textbf{LightGCN} & \textbf{FedMF} & \textbf{FedNCF} & \textbf{FedFast} & \textbf{PFedRec}& \textbf{CoLR} & \textbf{GPFedRec} & \textbf{FedRAP} & \textbf{FedCIA} & \textbf{FedEM} \\
\midrule
\multirow{2}{*}{\textbf{Filmtrust}} & HR@5 & 0.6287 & 0.6537 & 0.7385 & 0.6267 & 0.6347 & 0.6248 & 0.6367 & 0.6018 & 0.6447 & 0.6766 & \underline{0.7774} & \textbf{0.8353} \\
& NDCG@5 & 0.5128 & 0.5374 & 0.6338 & 0.5188 & 0.5257 & 0.4826 & 0.5271 & 0.4991 & 0.5298 & 0.5046 & \underline{0.6744} & \textbf{0.7511} \\
\midrule
\multirow{2}{*}{\textbf{ML-100K}} & HR@5 & 0.4719 & 0.4062 & 0.7063 & 0.3277 & 0.2694 & 0.3181 & 0.4772 & 0.3383 & 0.5133 & 0.8367 & \underline{0.8717} & \textbf{0.9745} \\
& NDCG@5 & 0.3177 & 0.2755 & 0.5395 & 0.2201 & 0.1794 & 0.2213 & 0.3230 & 0.2266 & 0.3461 & \underline{0.7779} & 0.7479 & \textbf{0.9359} \\
\midrule
\multirow{2}{*}{\textbf{ML-1M}} & HR@5 & 0.4209 & 0.4121 & 0.6831 & 0.3353 & 0.2704 & 0.2921 & 0.4866 & 0.3060 & 0.4957 & 0.8086 & \underline{0.8411} & \textbf{0.9035} \\
& NDCG@5 & 0.2779 & 0.2750 & 0.5053 & 0.2243 & 0.1795 & 0.1940 & 0.3296 & 0.2002 & 0.3382 & 0.7489 & \underline{0.7695} & \textbf{0.8238} \\
\midrule
\multirow{2}{*}{\textbf{LastFM-2K}} & HR@5 & 0.7628 & 0.6879 & 0.7833 & 0.4846 & 0.3987 & 0.3932 & 0.7226 & 0.4240 & \textbf{0.7415} & 0.5950 & 0.6690 & \underline{0.7360} \\
& NDCG@5 & 0.5896 & 0.5687 & 0.6687 & 0.3613 & 0.2985 & 0.2824 & \underline{0.6628} & 0.3230 & 0.6508 & 0.5842 & 0.6468 & \textbf{0.6967} \\
\bottomrule
\end{tabular}
}
\caption{Performance comparison on four datasets, reported by HR@5 and NDCG@5. CenRec and FedRec represent centralized and federated recommendation methods, respectively. The best FedRec results are bold and the second ones are underlined.}
\label{tab:exp_main_}
\end{table*}

\subsection{Compared Baselines} 
We comprehensively compare FedEM with SOTA centralized methods as follows:\\
\textbf{Matrix Factorization (MF)}~\cite{koren2009matrix}: A foundational recommendation approach that decomposes the user-item interaction matrix into latent user and item embeddings, effectively capturing their respective features within a shared vector space.\\
\textbf{Neural Collaborative Filtering (NCF)}~\cite{he2017neural}: Extending traditional embedding-based methods, NCF employs a multi-layer perceptron (MLP) to model complex, non-linear interactions between users and items based on their latent representations.\\
\textbf{LightGCN}~\cite{he2020lightgcn}: A graph-based collaborative filtering method that simplifies traditional GCNs by removing non-linear transformations and feature combinations, enabling efficient learning of embeddings through linear propagation over the user-item interaction graph.

\noindent Besides, we compare FedEM with some SOTA federated methods, such as:\\
\textbf{FedMF}~\cite{chai2020secure}: A federated version of MF, where user embeddings are updated locally, and encrypted item gradients are shared with the server for aggregation. In our experiments, we use its unencrypted variant to focus on performance evaluation.\\
\textbf{FedNCF}~\cite{perifanis2022federated}: It extends NCF to the federated setting by updating user embeddings locally, while both item embeddings and the MLP components are uploaded for global aggregation on the server.\\
\textbf{FedFast}~\cite{muhammad2020fedfast}: It groups users into clusters according to their profile similarity, and applies cluster-specific sampling and aggregation strategies to enhance training efficiency.\\
\textbf{PFedRec}~\cite{zhang2023dual}: It utilizes a dual personalization mechanism that captures user preferences with a personalized score function and fine-grained personalization on item embeddings.\\
\textbf{CoLR}~\cite{nguyen2024towards}: It decomposes local updates into a low-rank matrix, where only a subset of parameters is trained and transmitted, while the remaining part is fixed and shared. This approach significantly reduces communication overhead while preserving model performance.\\
\textbf{GPFedRec}~\cite{zhang2024gpfedrec}: It proposes a graph-based aggregation strategy that captures inter-client relationships to guide global learning of user-specific item embeddings, thereby enhancing local personalization.\\
\textbf{FedRAP}~\cite{li2023federated}: It considers both the global view and user-specific view of items by applying an additive model to item embedding. This method requires updating the local model and the additive one for each client.\\
\textbf{FedCIA}~\cite{han2025fedcia}: It proposes a model-free aggregation method that uploads and aggregates client-side item similarity matrices, achieving better information retention and stronger personalization for each client.

\subsection{Implementation Details} 
In the experiments, we randomly sample $N=4$ negative instances for each positive sample, following prior works~\cite{he2017neural,zhang2023dual}. To ensure fair comparison, we set the embedding dimension $d=16$ and batch size $B=256$ for all methods, while other hyperparameter settings for baselines follow their default/optimal implementations. For centralized methods, the total number of training epochs is set to 100 to ensure convergence. For federated methods, we set the number of global rounds $R=100$ and sample $100\%$ of clients in each round to perform global aggregation for generality. Notably, for FedFast, we group clients into 5 clusters for stable training. For CoLR, the rank of decomposed updates is set to 4.

For our method, we adopt the same basic configurations as the backbone FedMF, \textit{i.e.}, a learning rate $\eta = 0.1$ for user and item embeddings, and $E=10$ local training epochs with the Adam optimizer. As for the additional hyperparameter $\alpha$, we perform a grid search within the range $[0.5, 1.5]$ with a step size of 0.1. By default, the adapter is implemented as a 3-hidden-layer MLP with the structure $[32, 16, 8, 1]$. We also explore MLPs with different depths, including $[32, 1]$, $[32, 16, 1]$, $[32, 32, 16, 8, 1]$, and $[32, 64, 32, 16, 8, 1]$. The learning rate $\beta$ for the adapter is set to be consistent with the embedding learning rate $\eta$, \textit{i.e.}, $\beta=0.1$. All experiments are conducted on a server equipped with four NVIDIA RTX A5000 GPUs, each with 24GB of memory.

\begin{figure*}[ht]
\centering
\includegraphics[width=0.99\textwidth]{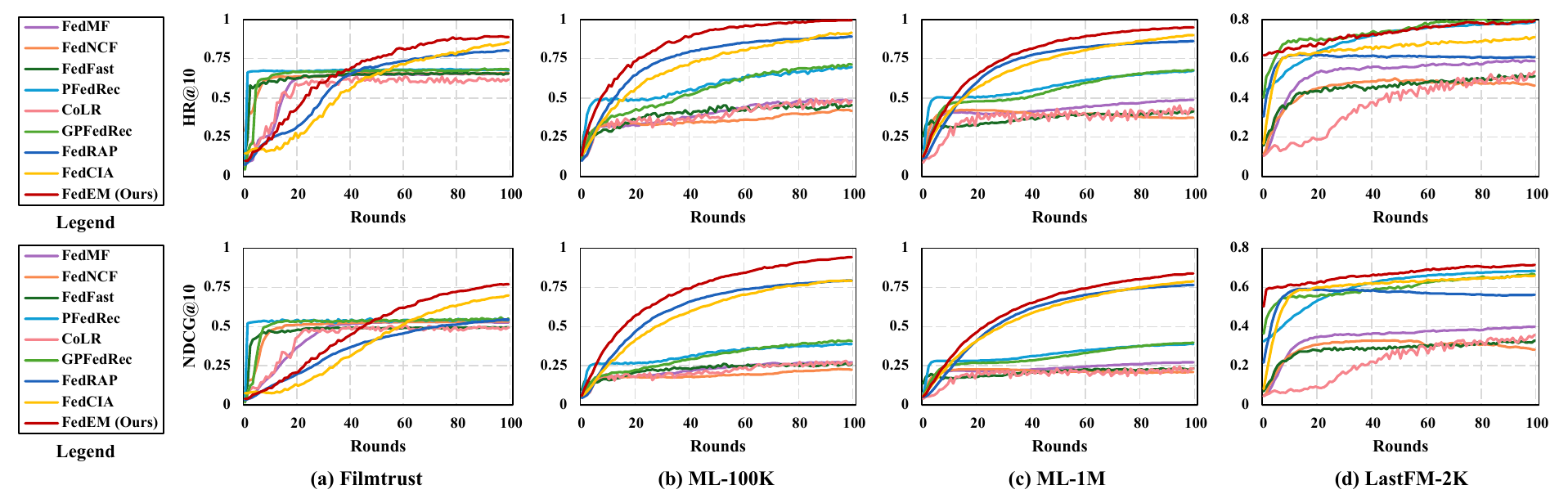}
\caption{Convergence comparison evaluated by HR@10 and NDCG@10. The horizontal axis is the federated rounds.}
\label{pic:convergence}
\end{figure*}

\section{E~~~~Extensive Experiment Results}
\subsection{More Baseline Comparison}
\textbf{Performance.} Furthermore, we present additional performance comparisons on HR@5 and NDCG@5, as illustrated in Figure~\ref{tab:exp_main_}. Our method outperforms both SOTA centralized and federated methods, which is consistent with the evaluation results on HR@10 and NDCG@10. In particular, on relatively dense datasets such as ML-100K and ML-1M, our method significantly outperforms the strongest baselines, with HR@5 exceeding 0.9. For extremely sparse datasets like LastFM-2K, existing federated methods generally underperform compared to centralized approaches such as MF and LightGCN, due to insufficient local interaction data on clients that leads to underfitting. Nevertheless, our method still slightly surpasses centralized models on the NDCG@5 and NDCG@10 metrics, demonstrating its effectiveness even under severe data sparsity. The superior performance of our method stems from identifying and addressing a fundamental issue in FR, rather than relying on heuristic designs. Moreover, the adopted model merging strategy is both theoretically grounded and empirically validated.

\noindent \textbf{Convergence.} We further compare the convergence of different methods, with results on four datasets shown in Figure~\ref{pic:convergence}. Compared to other methods: (1) FedEM demonstrates faster convergence in the early stages, especially on ML-100K, ML-1M, and LastFM-2K, nearly outperforming all baselines. (2) FedMF achieves more stable and superior performance in the later stages. Our method exhibits even smoother final convergence with minimal fluctuations, particularly noticeable on ML-100K and LastFM-2K. Moreover, its performance surpasses that of all other methods across all datasets. These advantages stem from the elastic merging scheme, which preserves local information from the previous round, effectively preventing the loss of both optimization and personalization information, thereby enhancing convergence stability and overall performance.

\subsection{More Ablation Study}

\textbf{Visualization of the vectors $\mathbf{w}$}. On the server side, we generate a weight vector $\mathbf{w}_u \in \mathbb{R}^n$ for each client $u$ to guide tailored aggregation. By stacking all the weight vectors, we obtain the aggregation matrix $\mathbf{W} = \{ \mathbf{w}_u \mid u \in \mathcal{U} \} \in \mathbb{R}^{n \times n}$. Here, we randomly select 10 clients in the same round on ML-100K and visualize the aggregation matrix $\mathbf{W}$. As shown in Figure~\ref{pic:vis_matrix}, different clients are assigned customized aggregation weights. The traditional fixed-weight aggregation, such as FedMF, assigns weights based on the relative amount of data, ensuring that clients with more training data have a larger influence in the aggregation. While this can incorporate more collaborative information, it may degrade the performance of the aggregated model when those dominant clients exhibit preferences dissimilar to others, as analyzed in Lemma~\ref{lemma:aggregation_bottleneck_}. To this end, we incorporate similarity-based weights to aggregate clients with similar preferences, thereby mitigating the negative impact of client heterogeneity and promoting more effective collaboration. This design further alleviates the performance bottleneck in global aggregation.

\begin{figure}[htbp]
\centering
\includegraphics[width=0.99\columnwidth]{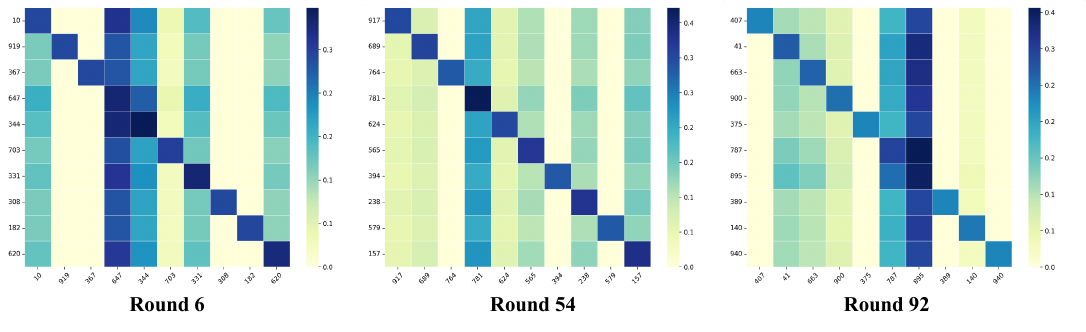} 
\caption{Visualization of matrices $\mathbf{W}$ for global aggregation of several randomly selected clients in different rounds.}
\label{pic:vis_matrix}
\end{figure}

\noindent \textbf{Visualization of the vectors $\bm{\rho}$}. We randomly select several clients on ML-100K in the same round and visualize their weight vectors $\bm{\rho}$ for elastic merging. For brevity, only the weights associated with the first 30 items are displayed, as shown in Figure~\ref{pic:vis_vector}. From the visualization perspective, FedEM takes into account the heterogeneity among clients, elastically yielding tailored $\bm{\rho}$ to coordinate the global and local models for each client. Additionally, it finely assigns different weights to each item within the same client, optimally exploiting both collaborative and personalized information. The proposed elastic merging scheme leverages the off-the-shelf local model to effectively compensate for the loss of local information introduced by global aggregation, thus breaking the performance bottleneck in FR.

\begin{figure}[htbp]
\centering
\includegraphics[width=0.99\columnwidth]{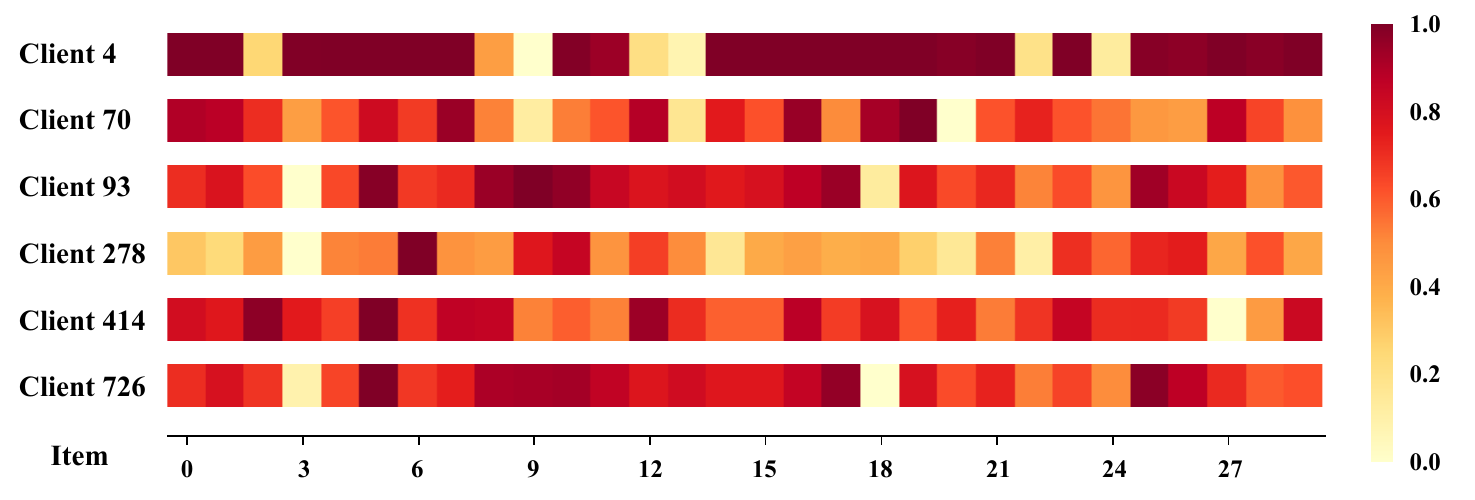} 
\caption{Visualization of vectors $\bm{\rho}$ for elastic merging of several randomly selected clients in the same round.}
\label{pic:vis_vector}
\end{figure}

\subsection{More Hyperparameter Analysis}
\textbf{Similarity Coefficient.} During global aggregation, we extend the fixed-weight strategy by additionally incorporating model similarity to generate client-specific aggregation weights, which are modulated by a similarity coefficient $\alpha$. As shown in Figure~\ref{pic:param_alpha}, we find that setting $\alpha$ around 1, \textit{i.e.}, balancing relative data size and model similarity, can effectively improve overall model performance. This is because aggregating models from similar clients mitigates the inherent heterogeneity, thereby reducing the potential negative effects of aggregation. Specifically, the best results are obtained when $\alpha=0.9$ on FilmTrust, ML-1M, and LastFM-2K, and when $\alpha=1.1$ on ML-100K.

\begin{figure}[ht]
\centering
\includegraphics[width=0.99\columnwidth]{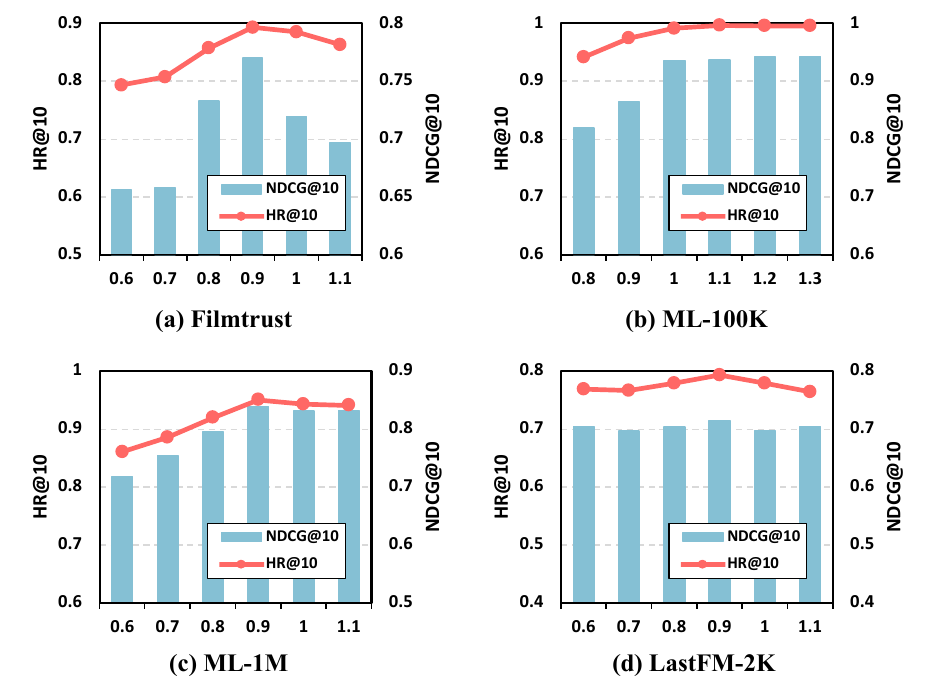}
\caption{Performance under different similarity coefficients $\alpha$ for global aggregation.}
\label{pic:param_alpha}
\end{figure}

\subsection{More Privacy Analysis}
Our method can be integrated with local differential privacy (LDP) to further enhance privacy protection. We evaluate the model performance under different noise strengths $\delta$, as shown in Table~\ref{tab:exp_ldp_}. As expected, the performance degrades as $\delta$ increases, indicating stronger privacy guarantees. Notably, FedEM remains robust when $\delta \in [0.1, 0.3]$, and even under stronger privacy constraints, \textit{e.g.}, $\delta = 0.5$, it still demonstrates competitive performance. In practice, to balance privacy and utility, we set the noise strength to $\delta = 0.3$.

\begin{table}[ht]
\centering
\resizebox{1.0\columnwidth}{!}{
\Large
\begin{tabular}{lcccccccc}
\toprule
\multirow{2}{*}{\textbf{$\delta$}}& \multicolumn{2}{c}{\textbf{Filmtrust}} & \multicolumn{2}{c}{\textbf{ML-100K}} & \multicolumn{2}{c}{\textbf{ML-1M}} & \multicolumn{2}{c}{\textbf{LastFM-2K}} \\
\cmidrule(lr){2-3}\cmidrule(lr){4-5}\cmidrule(lr){6-7}\cmidrule(lr){8-9}
& \large HR@10 & \large NDCG@10 & \large HR@10 & \large NDCG@10 & \large HR@10 & \large NDCG@10 & \large HR@10 & \large NDCG@10 \\
\midrule
0 & 0.8932 & 0.7701 & 0.9958 & 0.9427 & 0.9507 & 0.8392 & 0.7935 & 0.7151 \\
0.1 & 0.8892 & 0.7617 & 0.9936 & 0.9405 & 0.9490 & 0.8337 & 0.7817 & 0.7047 \\
0.2 & 0.8832 & 0.7459 & 0.9947 & 0.9319 & 0.9434 & 0.8249 & 0.7770 & 0.6999 \\
0.3 & 0.8802 & 0.7353 & 0.9936 & 0.9399 & 0.9392 & 0.8193 & 0.7707 & 0.6930 \\
0.4 & 0.8782 & 0.7194 & 0.9947 & 0.9342 & 0.9328 & 0.8122 & 0.7644 & 0.6857 \\
0.5 & 0.8643 & 0.7025 & 0.9905 & 0.9143 & 0.9308 & 0.8075 & 0.7533 & 0.6778 \\
\bottomrule
\end{tabular}
}
\caption{Performance of applying LDP to our method with different noise strength $\delta$.}
\label{tab:exp_ldp_}
\end{table}

\end{document}